\documentclass[journal]{IEEEtran}
\usepackage{amsmath,amsfonts}
\usepackage{algorithmic}
\usepackage{algorithm}
\usepackage{array}
\usepackage{bm}
\usepackage{textcomp}
\usepackage{stfloats}
\usepackage{url}
\usepackage{verbatim}
\usepackage{graphicx} 
\usepackage{cite}
\usepackage{amssymb}
\usepackage{CJK}
\usepackage{indentfirst}
\usepackage{cases}
\usepackage{setspace}
\usepackage{hyperref}
\usepackage{float}
\usepackage{titlecaps}
\usepackage{acronym}

\acrodef{isac}[ISAC]{integrated sensing and communications}
\acrodef{sinr}[SINR]{signal-to-interference-plus-noise ratio}
\acrodef{snr}[SNR]{signal-to-noise ratio}
\acrodef{crb}[CRB]{Cram{\'e}r-Rao Bound}
\acrodef{dof}[DoF]{degrees of freedom}
\acrodef{miso}[MISO]{ multiple-input single-output}
\acrodef{mimo}[MIMO]{multiple-input and multiple-output}
\acrodef{mu-mimo}[MU-MIMO]{multi-user multiple-input and multiple-output}
\acrodef{mi}[MI]{mutual information}
\acrodef{pwm}[PWM]{planar-wave model}
\acrodef{swm}[SWM]{spherical-wave model}
\acrodef{nusw}[NUSW]{non-uniform spherical-wave}
\acrodef{hspm}[HSPM]{hybrid spherical and planar wavefront model}
\acrodef{aoa}[AoA]{angle of arrival}
\acrodef{aod}[AoD]{angle of departure}
\acrodef{dod}[DoD]{direction of departure}
\acrodef{doa}[DoA]{direction of arrival}
\acrodef{hbf}[HBF]{hybrid beamforming}
\acrodef{los}[LoS]{line-of-sight}
\acrodef{nlos}[NLoS]{non-line-of-sight}
\acrodef{awgn}[AWGN]{additive white Gaussian noise}
\acrodef{rcs}[RCS]{radar cross section}
\acrodef{fim}[FIM]{Fisher’s information matrix}
\acrodef{scnr}[SCNR]{signal-to-clutter-plus-noise ratio}
\acrodef{mvdr}[MVDR]{minimum variance distortionless response}
\acrodef{svd}[SVD]{singular-value-decomlocation}
\acrodef{qmp}[QMP]{quadratic matrix programming}
\acrodef{sdr}[SDR]{semidefinite relaxation}
\acrodef{sdp}[SDP]{semidefinite programming}
\acrodef{kkt}[KKT]{Karush-Kuhn-Tucker}
\acrodef{music}[MUSIC]{mutiple signal 
classification}
\acrodef{xl}[XL]{extremely
large-scale}
\acrodef{mmwave}[mmWave]{millimeter wave}
\acrodef{thz}[THz]{terahertz}
\acrodef{xl-mimo}[XL-MIMO]{extremely
large-scale multiple-input and multiple-output}
\acrodef{wsms}[WSMS]{widely-spaced multi-subarray}
\acrodef{rf}[RF]{radio frequency}
\acrodef{bs}[BS]{base station}
\acrodef{ula}[ULA]{uniform linear array}
\acrodef{cpi}[CPI]{coherent processing interval}
\acrodef{doa}[DoA]{direction-of-arrival}
\acrodef{pri}[PRI]{pulse repetition interval}
\acrodef{sigw}[SIGW]{synchronous gridless weighted }
\acrodef{rmse}[RMSE]{root mean square error}
\acrodef{somp}[SOMP]{simultaneous orthogonal matching pursuit}
\acrodef{xl-mimo}[XL-MIMO]{extremely
large-scale multiple-input and multiple-output}
\acrodef{ris}[RIS]{reconfigurable intelligent surfaces}
\acrodef{pdf}[PDF]{probability density function}
\acrodef{mmse}[MMSE]{minimum mean-square error}
\acrodef{map}[MAP]{maximum a posteriori}
\acrodef{spa}[SPA]{sum-product algorithm}
\acrodef{bp}[BP]{Belief propagation}
\acrodef{amp}[AMP]{approximate message passing}
\acrodef{gamp}[GAMP]{generalized approximate message passing }
\acrodef{vm}[VM]{von Mises}
\acrodef{mvalse}[MVALSE]{multidimensional variational line spectrum estimation}
\acrodef{kld}[KLD]{Kullback-Leibler  divergence}
\acrodef{elbo}[ELBO]{evidence lower bound}
\acrodef{ep}[EP]{expectation propagation}
\acrodef{fim}[FIM]{Fisher information matrix}
\acrodef{mle}[MLE]{maximum likelihood estimation}
\acrodef{ml}[ML]{maximum likelihood}
\acrodef{6g}[6G]{sixth generation}
\acrodef{vmp}[VMP]{variational message passing}
\acrodef{vbi}[VBI]{variational Bayesian inference}
\acrodef{cavi}[CAVI]{coordinate ascent variational inference}
\acrodef{t-r}[T-R]{transmit-receive}

\ifCLASSOPTIONcompsoc
\usepackage[caption=false,font=normalsize,labelfont=sf,textfont=sf]{subfig}
\else
\usepackage[caption=false,font=footnotesize]{subfig}
\fi

\hypersetup{colorlinks=true}
\def\BibTeX{{\rm B\kern-.05em{\sc i\kern-.025em b}\kern-.08em
		T\kern-.1667em\lower.7ex\hbox{E}\kern-.125emX}}

\newtheorem{remark}{\textbf{Remark}}
\newtheorem{theorem}{\textbf{Theorem}}

\newtheorem{lemma}{\textbf{Lemma}}

\newtheorem{corollary}{\textbf{Corollary}}

\newtheorem{proposition}{\textbf{Proposition}}
\newtheorem{definition}{\textbf{Definition}}
\newenvironment{proof}{{\indent \indent \it Proof:\quad}}{\hfill $\blacksquare$\par}

\makeatletter

\newcommand{\Rmnum}[1]{\expandafter\@slowromancap\romannumeral #1@}
\makeatother

\usepackage{framed} 

\usepackage{cancel}

\begin{document}	
	\title{Near-Field  Motion Parameter Estimation: A Variational Bayesian  Approach}

	\author{Chunwei~Meng, Zhaolin~Wang, Zhiqing~Wei, Yuanwei~Liu,  Zhiyong~Feng
		\thanks{C. Meng, Z. Wei, and Z. Feng are with the Key Laboratory of Universal Wireless Communications, Ministry of Education, Beijing University of Posts and Telecommunications, Beijing 100876, China (e-mail: \{mengchunwei, weizhiqing, fengzy\}@bupt.edu.cn).
  Z. Wang is with the School of Electronic Engineering and Computer Science, Queen Mary University of London, London E1 4NS, U.K. (e-mail: zhaolin.wang@qmul.ac.uk).
  Y. Liu is with the Department of Electrical and Electronic Engineering, the
University of Hong Kong, Hong Kong, China (e-mail: yuanwei@hku.hk).
		}
		
	}
	\maketitle
	
	\begin{abstract}
A near-field motion parameter estimation method is proposed. 
In contract to far-field sensing systems, the near-field sensing system leverages spherical-wave characteristics to enable full-vector location and velocity estimation.
 Despite promising advantages, the near-field sensing system faces a significant challenge, where location and velocity parameters are intricately coupled within the signal.
     To address this challenge, a novel subarray-based variational message passing (VMP) method is proposed for near-field joint  location and velocity estimation.
First, a factor graph representation is introduced, employing subarray-level directional and Doppler parameters as intermediate variables to decouple the complex location-velocity dependencies.
     Based on this, the variational Bayesian inference is employed to obtain closed-form posterior distributions of subarray-level parameters.
     Subsequently, the message passing technique is employed, enabling tractable computation of location and velocity marginal distributions.
Two implementation strategies are proposed: 1) System-level fusion that aggregates all subarray posteriors for centralized estimation, or 2) Subarray-level fusion where locally processed estimates from subarrays are fused through Guassian product rule.
     Cramér-Rao bounds for location and velocity estimation are derived, providing theoretical performance limits. 
    Numerical results demonstrate that the proposed VMP method outperforms existing approaches while achieving a magnitude lower complexity. 
   Specifically, the proposed VMP method achieves centimeter-level location accuracy and sub-m/s velocity accuracy. 
    It also demonstrates robust performance for high-mobility targets, making the proposed VMP method suitable for real-time near-field sensing and communication applications.

	\end{abstract}
	

\section{Introduction}
The emergence of novel applications in the  \ac{6g} wireless systems, such as smart manufacturing and intelligent transportation systems, necessitates high-capacity communications and high-precision sensing capabilities\cite{liu2022isac,dong2024sensing,wei2023isac,liu2023tutorial,wang2023nearisac}.
The demands of these applications make the \ac{xl-mimo} and \ac{mmwave}/\ac{thz} technologies key enablers and trends for \ac{6g} wireless systems, offering dramatically improved spectral efficiency and spatial resolution.
The combination of the two technologies brings paradigm shift in the electromagnetic characteristics of wireless environments\cite{liu2023tutorial}.
Specifically, the electromagnetic region of the \ac{bs} is divided into near-field and far-field regions by the Rayleigh distance. 
As the antenna array size increases and the operating frequency increases, the Rayleigh distance extends to hundreds of meters, making it more likely that the \ac{bs}  will carry out communications and sensing in the near-field region\cite{cong2023near}.

The spherical-wave characteristics of near-field electromagnetic propagation introduce new opportunities for both communication and sensing systems\cite{liu2024nearfield}.
The potential benefits of near-field communications have been extensively investigated in numerous studies, encompassing spatial multiplexing, inter-user interference management, location division multiple access\cite{cong2023near,liu2024nearfield,wang2024tutorial6g}.
For near-field sensing, the estimation of target parameters, both static and dynamic, is a crucial research problem\cite{liu2025sensingen,wang2024rethink,yang2023enhancing,yang2023pfboundsubarray}. 
Current research in near-field sensing predominantly focuses on static target localization, leveraging the spatial resolution advantages of spherical wavefronts for enhanced range/angle estimation with reduced time-frequency resources\cite{yang2023pfboundsubarray,yangs2023xlris,dehkordi2023multistatic,Sakhnini2022nfcoherent}.
However, the paradigm shifts significantly when considering the moving target. 
Traditional far-field systems using compact arrays can only resolve radial velocity components, while distributed multistatic configurations enable full velocity vector estimation through spatial diversity, albeit at the cost of increased hardware complexity and synchronization requirements\cite{Sakhnini2022nfcoherent,giovan2024pfbound,wang2024velocity}. 
In \cite{wang2024velocity}, the concept of near-field velocity sensing was introduced, enabling a monostatic radar operating in the near-field region to obtain the complete two-dimensional velocity of the target, thus avoiding the requirements for synchronization and the hardware costs associated with multistatic radar systems.

Despite the significant benefits of near-field operations for velocity estimation, research on this topic is still limited.
The existing literature can be categorized into two main approaches: filter-based tracking methods and snapshot estimation methods\cite{jiang2024near,Rahal2024ris,wang2024velocity}.
In the former category,  the authors of\cite{jiang2024near} proposed a novel near-field sensing enabled predictive beamforming framework that leverages both single-\ac{cpi} estimation and multi-\ac{cpi} filtering techniques to estimate the target's location and velocity.
Nevertheless, this approach relies on the availability of continuous \ac{cpi}s, which may not be feasible for scenarios requiring rapid initial target detection, e.g., the scenario contains newly appearing targets.
For the latter category, the authors of \cite{wang2024velocity} investigated the maximum likelihood velocity estimator and its adoption in predictive beamforming applications, assuming known target locations.
The reconfigurable intelligent surfaces aided single-snapshot near-field sensing algorithm introduced in \cite{Rahal2024ris} decouples the location and velocity estimation utilizing the static target assumption and linear approximations. 
However, research on near-field joint location-velocity estimation under unknown initial conditions remains  scarce.


The transition to near-field sensing in \ac{xl-mimo} systems enables reconstructing complete velocity vectors via spherical wavefronts with spatially diverse path projections.
However, this advancement faces two critical challenges: 1) The antenna-specific spherical wave curvature induces strong nonlinear coupling between location and velocity parameters, creating complex interdependencies that complicate joint location-velocity estimation; 2) The resulting four-dimensional estimation space imposes prohibitive computational burdens on conventional methods, with complexity scaling combinatorially with array dimensions and observation duration.
To overcome these challenges, we develop a novel subarray-based \ac{vmp} algorithm for  near-field joint location and velocity estimation with computational efficiency.  
The main contributions of this paper are summarized as follows:
 \begin{itemize}
      \item       We propose a novel processing paradigm for near-field joint location and velocity sensing.
    By partitioning the arrays into subarrays and applying a piecewise-far-field model,  each \ac{t-r} subarray pair operates as a bistatic radar.
   This decouples location and velocity estimation by enabling independent estimation of subarray-level direction parameters and Doppler shifts.
    Besides, the  spherical wavefront propagation across subarrays provides multiperspective spatial diversity, enabling simultaneous 2D location and velocity estimation.

\item      We introduce \ac{doa}, \ac{dod}, and bistatic Doppler shift as subarray-level intermediate variables for each  \ac{t-r} subarray pair.
Based on this, we establish a probabilistic graphical model that resolves intricate location-velocity dependencies through geometric conditional relationships, enabling efficient separate inference stages. 
We then develop a \ac{vbi} framework to derive closed-form posterior distributions of subarray-level parameters and estimate crucial hyperparameters, enabling efficient parallel processing at the subarray level.

\item    We propose a message passing-based algorithm to achieve efficient joint location-velocity estimation by
systematically fusing subarray measurements.
The proposed algorithm can operate in two complementary modes:
1) System-level fusion: Performs centralized estimation by integrating all subarray posteriors into a system-level optimization problem solved via gradient descent.
2) Subarray-level fusion: Performs distributed estimation by obtaining subarray-level location and velocity estimates through Laplace approximations, followed by the fusion of Gaussian mixtures  to derive the final estimates.

       \item 
    We evaluate the performance of the proposed method with existing approaches under different scenarios. 
     Numerical results demonstrate that the proposed method achieves centimeter-level localization and sub-m/s velocity accuracy,  while achieving an order-of-magnitude reduction in complexity.
  The proposed system-level fusion achieves higher estimation accuracy, while the subarray-level method offers lower computational complexity through distributed processing.
     Moreover, the proposed method exhibits robust performance for the high-velocity target, making it suitable for real-time applications.
 \end{itemize}

The rest of this paper is organized as follows.
Section~\ref{section_system_model} presents the system model,  channel model, and problem formulation.
Section~\ref{section_Probability_model} establishes the probabilistic model. 
Section~\ref{section_posteri_pdf_sub}  develops the variational inference framework for subarray-level parameter estimation.
Section~\ref{section_algorithm} develops the message passing-based algorithm for location and velocity estimation. 
Section~\ref{section_performance_analysis}  derives the location and velocity \ac{crb}s.
Section~\ref{section_simulation} presents simulation results.
Finally, Section~\ref{section_conclusion} concludes the paper.

\textit{Notation:} Vectors and matrices are denoted by boldface lowercase and uppercase letters. $(\cdot)^T$, $(\cdot)^H$, $(\cdot)^{-1}$, and $\|\cdot\|_F$ represent transpose, conjugate transpose, matrix inversion, and Frobenius norm, respectively. $|\cdot|$ and $\|\cdot\|$ denote the absolute value and Euclidean norm, respectively. $\mathcal{CN}(\cdot,\cdot)$, $\mathcal{M}(\cdot,\cdot)$, $\Re\{\cdot\}$, $\Im\{\cdot\}$, $\mathbb{E}\{\cdot\}$, and $\text{tr}\{\cdot\}$ denote complex Gaussian distribution, von Mises distribution, real part operator, imaginary part operator, expectation operator, and trace operator, respectively. $\otimes$, and $\succeq$ represent Kronecker product and positive semidefinite, respectively. $\mathbf{I}$ and $\mathbf{0}$ denote an identity matrix and a zero matrix. $\delta(\cdot)$ represents the Dirac delta function.

\section{System Model and Problem Formulation} \label{section_system_model} 
\begin{figure} [t]
\setlength{\abovecaptionskip}{-0.1cm}
		\centering
	\includegraphics[width=0.35\textwidth]{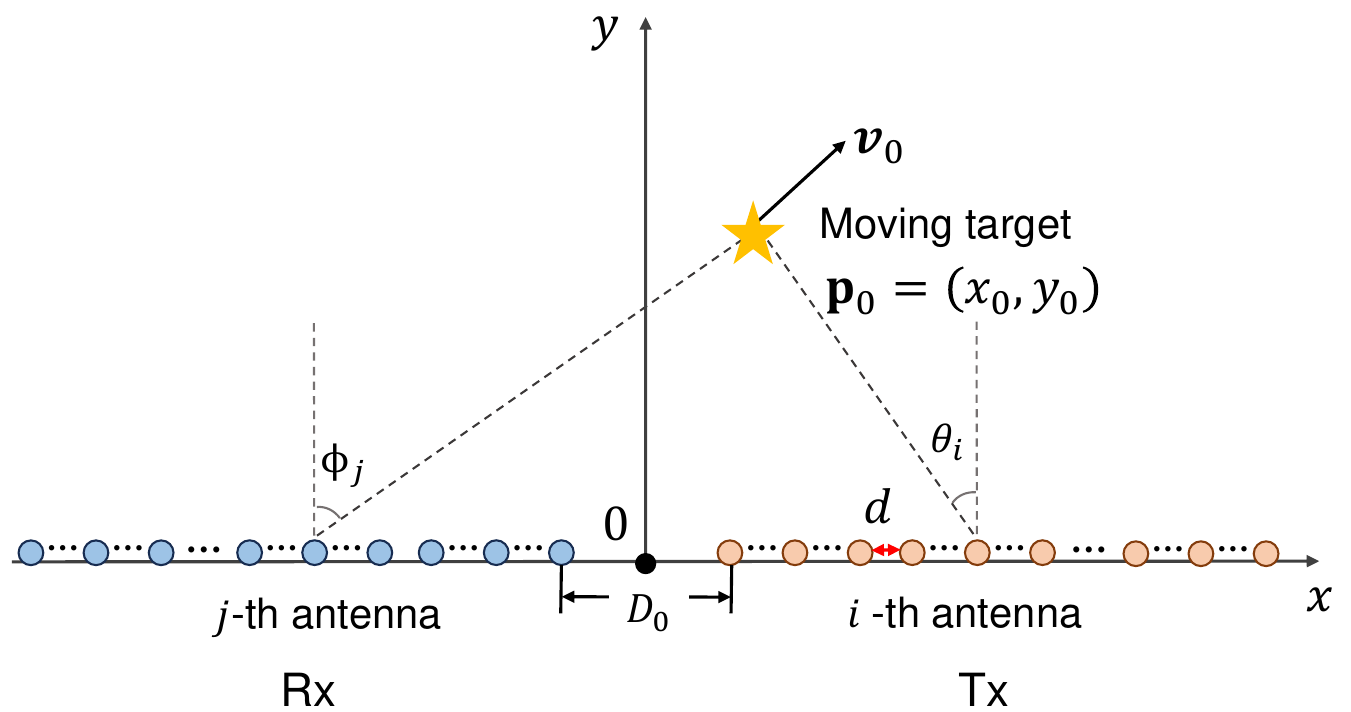}
		\caption{ {Illustration of the downlink near-field sensing scenario  where the \ac{bs} equipped with \ac{xl-mimo}s senses a moving target in the near-field region.}}
		\label{scenario}\vspace{-3mm}
	\end{figure}
    
 {As illustrated in Fig. \ref{scenario},}  the \ac{bs}  equipped with  transmit and receive XL-\ac{ula}s  actively senses a moving target.
The \ac{bs}'s transmit and receive arrays are deployed along the $x$-axis, symmetrically centered around the origin, and widely separated in space to suppress signal leakage from transmitter.
The total number of antenna elements in transmit and receive arrays are $N_t$ and $N_r$, respectively, with an inter-element spacing of $d = \frac{\lambda}{2}$, where $\lambda$ denotes the wavelength.
Assuming the distance between transmit and receive arrays is $D_0$,  the location of the $i$-th transmit antenna is given by $ \mathbf{p}_{i}^{t}=(\frac{D_0}{2} + (i-1)d, 0)$, and  the location of the $j$-th receive antenna is given by $ \mathbf{p}_{j}^{r}=(-\frac{D_0}{2} - (j-1)d, 0)$.
The boundary between the near-field and far-field regions of the \ac{bs} is determined by the Rayleigh distance, which is given by $R = \min\{\frac{2(N_t-1)^2d^2}{\lambda}, \frac{2(N_r-1)^2d^2}{\lambda}\}$.
Assume that the moving target is located in the near-field region of both arrays, with its initial location and velocity denoted by $\mathbf{p}_0 = (x_0, y_0)$ and $\mathbf{v}_0 = (v_x, v_y)$, respectively.

 Each antenna of the transmitter emits temporally orthogonal signals that impinge on the target, and the reflections are intercepted
by the receiver\cite{li2008mimo}.
The signal received at the $j$ -th receive antenna at time $t$ is given by
\begin{equation} \label{receive_sig_jth}
    \begin{aligned}
    {y}_j(t)\!=&\!\!\sum\nolimits_{i=1}^{N_t}\!\beta_{ij}e^{-j\frac{2\pi}
     {\lambda}\!\left[r_{T_i}\!(t)+r_{R_j}\!(t)\right]} s_i(t\!-\!{\tau}_{T_i}\!(t)\!-\!{\tau}_{R_j}\!(t))\!\!\\
     &+\!{n}_{ij}(t),
\end{aligned}
\end{equation}
where $\beta_{ij}$ denotes the channel gain, ${n}_{ij}(t)$  denotes the additive noise which is characterized as i.i.d. complex Gaussian with zero mean and variance ${\sigma}^2$.
 The time-variant propagation distances from the $i$-th transmit antenna to the target and from the $j$-th receive antenna to the target are represented by $r{T_i}(t)$ and $r_{R_j}(t)$, respectively. 
Specifically, we have $r_{T_i}(t) = r_{T_i}+v_i t$ and $r_{R_j}(t) = r_{R_j}+v_j t$, where $r_{T_i} = \|\mathbf{p}_i^t - \mathbf{p}_0\|$ and $r_{R_j} = \|\mathbf{p}_j^r - \mathbf{p}_0\|$ denote initial distances between the target and the $i$-th transmit antenna and the $j$-th receive antenna, respectively.
The terms $v_i$ and $v_j$ represent the radial velocity components obtained by projecting the target's velocity vector onto the lines connecting the target with the $i$-th transmit antenna and the $j$-th receive antenna, respectively.
Additionally, ${\tau}_{T_i}(t) = {r_{T_i}(t)}/{c}$ and ${\tau}_{R_j}(t) = {r_{R_j}(t)}/{c}$ are the time delays, where $c$ denotes the speed of light.
Following the radar range equation\cite{richards2010principles}, the channel gain $\beta_{ij}$ can be modeled as $\beta_{ij} = {\frac{\sqrt{{{P}_{T}}{{G}_{T}}G_R}{\alpha}_{\text{RCS}}}{{{r}_{{{T}_{i}}}}{{r}_{R_j}}}}$, where $P_T$ is the transmit signal power, $G_T$ represents the antenna gain of the transmit array,  $G_R$ represents the antenna gain of the receive array, ${\alpha}_{\text{RCS}}$ is the \ac{rcs} of the target, which is  modeled as a Swerling I-type fluctuation, remaining constant over the observation interval with ${\alpha}_{\text{RCS}} \sim \mathcal{CN}(0, \sigma_s^2)$\cite{skolnik1980introduction}.

Considering $L$ pulses within a \ac{cpi}, the slow-time sampling instants are given by $t_l = lT$ for $l=0,...,L-1$, where $T$ denotes the \ac{pri}. 
Under the common no-range-migration assumption within short \ac{cpi}\cite{richards2010principles}, the matched filter output at the $j$-th receive antenna for the $i$-th transmit antenna and $l$-th pulse is given by
\begin{equation}
Y_{ij}^{(l)} = \beta_{ij} \exp\Big(-j\frac{2\pi}{\lambda}\Phi_{ij}^{(l)}(\mathbf{p}_0,\mathbf{v}_0)\Big) + N_{ij}^{(l)},
\end{equation}
where the aggregate phase accumulation $\Phi_{ij}^{(l)}(\mathbf{p}_0,\mathbf{v}_0) \in \mathbb{R}$ combines spatial propagation delay and Doppler-induced phase modulation, i.e.,
\begin{equation}\label{phase_component}
\Phi_{ij}^{(l)}(\mathbf{p}_0,\mathbf{v}_0) \!\triangleq \!\underbrace{\|\mathbf{p}_i^t\!-\!\mathbf{p}_0\| \!+\! \|\mathbf{p}_j^r\!-\!\mathbf{p}_0\|}_{\text{Static path delay}}\! + \!\!\!\underbrace{\lambda t_l f_{ij}(\mathbf{p}_0,\mathbf{v}_0)}_{\text{Dynamic Doppler shift}}.
\end{equation}
Specifically, the Doppler shift $f_{ij}(\mathbf{p}_0,\mathbf{v}_0)$ for the $(i,j)$-th \ac{t-r} antenna pair is given by
\begin{equation}\label{f_ij_antena}
     \setlength{\abovedisplayskip}{2pt}
	\setlength{\belowdisplayskip}{2pt}
        {f}_{ij}(\mathbf{p}_0,\mathbf{v}_0)
        =\frac{1}{\lambda }\Big[ \frac{{[ \mathbf{p}_{i}^{t}-{{\mathbf{p}}_{0}}]^{T}}{{\mathbf{v}}_{0}}}{\| \mathbf{p}_{i}^{t}-{{\mathbf{p}}_{0}}\|}+\frac{{[ \mathbf{p}_{j}^{r}-{{\mathbf{p}}_{0}}]^{T}}{{\mathbf{v}}_{0}}}{\| \mathbf{p}_{j}^{r}-{{\mathbf{p}}_{0}}\|} \Big].
\end{equation}
Therefore, the maximum likelihood estimation  requires solving a non-convex optimization, i.e.,
\begin{equation} \label{original_problem_ml}
     \setlength{\abovedisplayskip}{3pt}
	\setlength{\belowdisplayskip}{3pt}
 \arg\min_{\mathbf{p}_0,\mathbf{v}_0} \sum_{i=1}^{N_t}\sum_{j=1}^{N_r}\sum_{l=1}^L \Big| {Y_{ij}^{(l)}} -\beta_{ij} e^{-j\frac{2\pi}{\lambda}\bm{{\Phi}}_{ij}^{(l)}(\mathbf{p}_0,\mathbf{v}_0)} \Big|^2.
\end{equation}

The phase function in \eqref{phase_component} induces inseparable nonlinear coupling between $\mathbf{p}_0$ and $\mathbf{v}_0$ through each antenna pair, which necessitates 4D joint estimation over $(x_0,y_0,v_x,v_y)$. 
This  renders the problem  \eqref{original_problem_ml}  computationally intractable due to the exponential complexity scaling with the parameter space dimensionality. 
To address this limitation, we propose a low-complexity variational Bayesian approach in the following sections, which enables subarray-level distributed processing.

\section{Probability Model of Joint Location and Velocity Estimation} \label{section_Probability_model}
 In this section, we establish the probability model for joint location and velocity estimation problem, leveraging the piecewise-far-field approximation of the near-field channel.

\subsection{Piecewise-Far-Field Channel Model}

To decouple the nonlinear interdependency between location and velocity in the phase term \eqref{phase_component} while maintaining acceptable accuracy, we employ a piecewise-far-field channel model to approximate the intricate near-field channel\cite{cui2024nearwide}.

It can be observed that the Rayleigh distance is proportional to the square of the number of antennas, implying that the far-field assumption becomes more accurate for a smaller number of antennas.
Consequently, in the piecewise-far-field channel model, we partition the entire transmit and receive arrays into $K_t\triangleq \frac{N_t}{M}$ and $K_r\triangleq\frac{N_r}{N_0}$ non-overlapping subarrays, respectively, where $M$ denotes the number of antennas per subarray. 
This partitioning creates subarrays with significantly reduced Rayleigh distances (typically meters), enabling far-field planar-wave approximations within individual subarrays while preserving near-field spherical-wave characteristics across subarrays.

Specifically, we select the first antenna of each subarray as the reference antenna to represent the subarray's location.
Thus, the position of the $m$-th transmit subarray is given by ${\mathbf{p}}_m^t =(x_m^t, y_m^t)=(\frac{D_0}{2} + (m-1)Md, 0)$, and the position of the $n$-th receive subarray is given by ${\mathbf{p}}_n^r =(x_n^r, y_n^r)=(-\frac{D_0}{2} - (n-1)Md, 0)$, where $m \in \mathcal{I}_t \triangleq \{1, 2, \dots, K_t\}$ for the transmit array and $n \in \mathcal{I}_r \triangleq \{1, 2, \dots, K_r\}$ for the receive array.
Based on the far-field assumption within each subarray, the \ac{dod} $\tilde{\theta}_m$ is defined as the azimuth angle between the target  and the $m$-th transmit subarray's reference antenna, i.e., $\tilde{\theta}_m \triangleq \arctan\left(\frac{x_0 - x_m^t}{y_0 - y_m^t}\right), \quad m \in \mathcal{I}_t$.
Similarly, the \ac{doa} $\tilde{\phi}_n$ characterizes the arrival angle from the target to the $n$-th receive subarray at $\mathbf{p}_n^r$, i.e., $\tilde{\phi}_n \triangleq \arctan\left(\frac{x_0 - x_n^r}{y_0 - y_n^r}\right), \quad n \in \mathcal{I}_r$.
Therefore, each \ac{t-r} subarray pair can be regarded as a traditional bistatic radar. 
The matched filter output of  the received signal for
the $(m, n)$-th \ac{t-r} subarray pair at the $l$-th pulse is given by
\begin{equation} \label{z_ij_l}
    \begin{aligned}
        \mathbf{z}_{mn}^{(l)}= {\beta}_{mn}{{e}^{-j2\pi {{f}_{mn}}{{t}_{l}}}}\mathbf{a}_r({\phi}_n)\otimes \mathbf{a}_t(\theta_m) + \mathbf{n}_{mn}^{(l)},
    \end{aligned}
\end{equation}
 {where ${\beta}_{mn} \triangleq \frac{\sqrt{{{P}_{T}}{{G}_{T}}G_R}\alpha_{\text{RCS}}}{r_{t,m}r_{r,n}}$, with $r_{t,m} \triangleq\| {\mathbf{p}}_{m}^{t}-{{\mathbf{p}}_{0}} \|$ and $r_{r,n} \triangleq \| {\mathbf{p}}_{n}^{r}-{{\mathbf{p}}_{0}}\|$}.
 The bistatic Doppler frequency shift associated with the $(m, n)$-th \ac{t-r} subarray pair is represented by ${f}_{mn}$, which is given by
\begin{equation}
     \setlength{\abovedisplayskip}{2pt}
	\setlength{\belowdisplayskip}{2pt}
        {f}_{mn}
       =-\frac{1}{\lambda }\Big[ \frac{{[ \mathbf{p}_{m}^{t}-{{\mathbf{p}}_{0}}]^{T}}{{\mathbf{v}}_{0}}}{\| \mathbf{p}_{m}^{t}-{{\mathbf{p}}_{0}}\|}+\frac{{[ \mathbf{p}_{n}^{r}-{{\mathbf{p}}_{0}}]^{T}}{{\mathbf{v}}_{0}}}{\| \mathbf{p}_{n}^{r}-{{\mathbf{p}}_{0}}\|} \Big].
\end{equation}
In addition, ${{\mathbf{a}}_{t}}({{\theta }_{m}})={{[1,{{e}^{j{{\theta }_{m}}}},\ldots ,{{e}^{j(M-1){{\theta }_{m}}}}]^{T}}}$ and ${{\mathbf{a}}_{r}}({{\phi }_{n}})={{[1,{{e}^{j{{\phi }_{n}}}},\ldots ,{{e}^{j(M-1){{\phi }_{n}}}}]^{T}}}$ denote the $m$-th transmit and $n$-th receive subarray steering vectors, respectively, with ${{\theta }_{m}}\triangleq {2\pi d }/{\lambda }\sin {{{\tilde{\theta }}}_{m}}$ and ${{\phi }_{n}}\triangleq {2\pi d}/{\lambda }\sin {{{\tilde{\phi }}}_{n}}$.
The post-matched-filter noise term $\mathbf{n}_{mn}^{(l)}$ preserves the complex Gaussian distribution\cite{richards2005fundamentals}.

By stacking the slow-time samples $\mathbf{z}_{mn}^{(l)}$ for $t_l = lT$, $l = 0, 1, \ldots, L-1$ as columns, we obtain that
\begin{equation}\label{stack_zij}
     \setlength{\abovedisplayskip}{2pt}
	\setlength{\belowdisplayskip}{2pt}
        \mathbf{z}_{mn} = {\beta}_{mn}\mathbf{a}_r({\phi}_n)\otimes \mathbf{a}_t(\theta_m)  \otimes \mathbf{d}({f}_{mn})\ + \mathbf{n}_{mn}, 
\end{equation}
where the complex channel gain ${\beta}_{mn}$  follows $\mathcal{CN}(0,\varsigma_{mn})$, and
\begin{equation}
     \setlength{\abovedisplayskip}{2pt}
	\setlength{\belowdisplayskip}{2pt}
    \mathbf{d}({f}_{mn}) \!=\! [1, e^{-j2\pi {f}_{mn}T}, \ldots, e^{-j2\pi {f}_{mn}{(L-1)T}}]^T.
\end{equation}

\subsection{Probabilistic Model and Factor Graph Representation}
Building upon the piecewise-far-field channel model, we now establish the probability model of the joint location and velocity estimation problem.
Based on the subarray signal model in \eqref{stack_zij}, the likelihood function of the received signal $ \mathbf{z}_{mn}$ given $\mathbf{p}_0$, $\mathbf{v}_0$ and ${\beta}_{mn}$ can be expressed as
\begin{equation}\label{likelihood_zmn}
     \setlength{\abovedisplayskip}{4pt}
	\setlength{\belowdisplayskip}{1pt}
\begin{aligned}
    & p( {{\mathbf{z}}_{mn}}| {{\mathbf{p}}_{0}},{{\mathbf{v}}_{0}},{{\beta }_{mn}})=\mathcal{C}\mathcal{N}( {{\mathbf{z}}_{mn}};{{\beta }_{mn}}\boldsymbol{\mu}_{mn},{{\sigma }}{\mathbf{I}_{{{M}^{2}}L}} ),\\
&=\frac{1}{{{(\pi {{\sigma }})}^{{{M}^{2}}L}}}\exp ( -\frac{1}{{{\sigma }}}\| {{\mathbf{z}}_{mn}}-{{\beta }_{mn}}{{\boldsymbol{\mu }}_{mn}}\|_{\text{2}}^{\text{2}} ),
\end{aligned}
\end{equation}
where $\boldsymbol{\mu}_{mn} \triangleq \mathbf{a}_r({\phi}_n)\otimes \mathbf{a}_t(\theta_m) \otimes \mathbf{d}({{f}_{mn}})$.
By denoting ${\boldsymbol{\beta}} \triangleq [\beta_{11}, \ldots, \beta_{mn},\ldots, \beta_{K_tK_r}]$, the joint \ac{pdf} is given by
\begin{equation} \label{origi_joint_distribution}
     \setlength{\abovedisplayskip}{4pt}
	\setlength{\belowdisplayskip}{3pt}
    \begin{aligned}
        p( \mathbf{z},{{\mathbf{p}}_{0}},{{\mathbf{v}}_{0}},  \boldsymbol{\beta}) =  &\prod\nolimits_{m=1}^{{{K}_{t}}}\prod\nolimits_{n=1}^{{{K}_{r}}}p({{\mathbf{z}}_{mn}}| {{\mathbf{p}}_{0}},{{\mathbf{v}}_{0}},\!{{\beta }_{mn}})\\
       & \times p( {{\beta }_{mn}})p\left( {{\mathbf{p}}_{0}}\right)p\left({{\mathbf{v}}_{0}}\right).
        \end{aligned}
\end{equation}

 Our aim is to find the posteriori distributions $p( {{\mathbf{p}}_{0}}| \mathbf{z})$ and $p( {{\mathbf{v}}_{0}}| \mathbf{z})$ and their estimates by using either the \ac{mmse} or \ac{map} principles.
 Following the Bayes’ theorem, the posterior distributions of  $\mathbf{p}_0$ and $\mathbf{v}_0$ are given by
\begin{subequations} \label{post_pdf_po_vo}
     \setlength{\abovedisplayskip}{3pt}
	\setlength{\belowdisplayskip}{3pt}
    \begin{align}
    \label{post_pdf_p0}
         & p( {{\mathbf{p}}_{0}}| \mathbf{z}  )=\int{\frac{p( \mathbf{z},{{\mathbf{p}}_{0}},{{\mathbf{v}}_{0}},\boldsymbol{\beta})}{p( \mathbf{z})}}d{{\mathbf{v}}_{0}}d{\boldsymbol{\beta}} , \\  \label{post_pdf_v0}
 & p( {{\mathbf{v}}_{0}}| \mathbf{z})=\int{\frac{p( \mathbf{z},{{\mathbf{p}}_{0}},{{\mathbf{v}}_{0}},\boldsymbol{\beta} )}{p( \mathbf{z} )}}d{{\mathbf{p}}_{0}}d{\boldsymbol{\beta}} .
    \end{align}
\end{subequations}

However, the coupled nature of $\mathbf{p}_0$, $\mathbf{v}_0$, and $\boldsymbol{\beta}$ in \eqref{origi_joint_distribution} renders the direct evaluation of the high-dimensional integrals in \eqref{post_pdf_po_vo} computationally prohibitive.
To resolve the intricate coupling, we introduce intermediate variables \(\phi_n\) (\ac{doa}), \(\theta_m\) (\ac{dod}), and \(f_{mn}\) (bistatic Doppler frequency) for each \ac{t-r} subarray pair.
By conditioning on these intermediate variables, the likelihood \(p(\mathbf{z}_{mn} | \mathbf{p}_0, \mathbf{v}_0, \beta_{mn})\) can be factorized into a cascade of simpler distributions: \(p(\mathbf{z}_{mn} | \beta_{mn}, \phi_n, \theta_m, f_{mn})\) and the geometric mappings \(p(\phi_n | \mathbf{p}_0)\), \(p(\theta_m | \mathbf{p}_0)\), and \(p(f_{mn} | \phi_n, \theta_m, \mathbf{v}_0)\), i.e.,
\begin{equation} \label{joint_pdf}
    \begin{aligned}
 & p( \mathbf{z},{{\mathbf{p}}_{0}},{{\mathbf{v}}_{0}}, {\boldsymbol{\beta}})  \!=\!{\prod\nolimits_{m=1}^{{{K}_{t}}}\prod\nolimits_{n=1}^{{{K}_{r}}}{p( {{\mathbf{z}}_{mn}} | {{\phi }_{n}},{{\theta }_{m}},{{f}_{mn}},{{\beta }_{mn}} )}}\times \\ 
 &  p( {{f}_{mn}} | {{\theta }_{m}},{{\phi }_{n}},{{\mathbf{v}}_{0}} ) p( {{\mathbf{v}}_{0}} )p( {{\phi }_{n}} | {{\mathbf{p}}_{0}})p( {{\theta }_{m}} | {{\mathbf{p}}_{0}} )p( {{\mathbf{p}}_{0}})p( {{\beta }_{mn}} ),
    \end{aligned}
\end{equation}
where
\begin{subequations} \label{geometry_relationship}
     \setlength{\abovedisplayskip}{2pt}
	\setlength{\belowdisplayskip}{2pt}
    \begin{align}
  & p( {{\phi }_{n}}| {{\mathbf{p}}_{0}} )=\delta ({{\phi }_{n}}-{2\pi d}/{\lambda }\mathbf{e}(\phi_n)^T{{\mathbf{e}}_{x}}), \\ 
 & p( {{\theta }_{m}}| {{\mathbf{p}}_{0}})=\delta ( {{\theta }_{m}}-{2\pi d}/{\lambda} \mathbf{e}(\theta_m)^T{{\mathbf{e}}_{x}}), \\ 
\label{fmn_geomitric_relation} 
 & p( {{f}_{mn}} | {{\phi }_{n}},{{\theta }_{m}},{{\mathbf{v}}_{0}} )\\ \notag
 &= \delta ( {{f}_{mn}} - {1}/{\lambda}[ \mathbf{v}_{0}^{T}\mathbf{e}(\phi_n)  + \mathbf{v}_{0}^{T}\mathbf{e}(\theta_m)]) \\ \notag
&= \!\delta(\! {{f}_{mn}}\!\!-\!\!{1}/{\lambda}[ v_x\!\cos\tilde{\theta}_m \!\!+\!  v_y\sin\tilde{\theta}_m \!\!+\! v_x\cos\tilde{\phi}_n \!\!+\! v_y\sin\tilde{\phi}_n ] ),
\end{align}
\end{subequations}
with 
\begin{equation}
     \setlength{\abovedisplayskip}{2pt}
	\setlength{\belowdisplayskip}{3pt}
    \mathbf{e}(\phi_n)\triangleq\frac{{( \mathbf{p}_{n}^{r}-{{\mathbf{p}}_{0}})}}{\left\|\mathbf{p}_{n}^{r}-{{\mathbf{p}}_{0}} \right\|},\ \ \mathbf{e}(\theta_m)\triangleq\frac{{{( \mathbf{p}_{m}^{t}-{{\mathbf{p}}_{0}})}}}{\left\| \mathbf{p}_{m}^{t}-{{\mathbf{p}}_{0}} \right\|}.
\end{equation}

The  factorization in \eqref{joint_pdf} can be naturally represented by a factor graph \cite{loe2004factorgraph}, a bipartite graphical model consisting of variable nodes for parameters and factor nodes encoding their probabilistic relationships as detailed in Table~\ref{table_factor_nodes}. 
As shown in Fig.~\ref{factor_graph}, this graphical model explicitly captures the dependencies between motion parameters ($\mathbf{p}_0,\mathbf{v}_0$) and subarray-level variables ($\phi_n,\theta_m,f_{mn}, \beta_{mn}$), enabling the application of \ac{vbi} (Section \ref{section_posteri_pdf_sub}) and message passing techniques (Section \ref{section_algorithm}) for efficient location and velocity estimation.


\begin{table}[t]
\centering
\caption{Factor Nodes in Fig.~\ref{factor_graph}}
\label{table_factor_nodes}
\begin{tabular}{|c|c|}
\hline
Factor Node & Factor Function \\
\hline
$p_{\phi_n}$ & $p(\phi_n|\mathbf{p}_0)$ \\
\hline
$p_{\theta_m}$ & $p(\theta_m|\mathbf{p}_0)$ \\
\hline
$p_{f_{mn}}$ & $p(f_{mn}|\phi_n, \theta_m, \mathbf{v}_0)$ \\
\hline
$p_{\beta_{mn}}$ & $p(\beta_{mn})$ \\
\hline
$p_{z_{mn}}$ & $p(\mathbf{z}_{mn}|\phi_n, \theta_m, f_{mn}, \beta_{mn})$ \\
\hline
\end{tabular}\vspace{-3mm}
\end{table}

\begin{figure} [t]
\setlength{\abovecaptionskip}{-0.1cm}
		\centering
	\includegraphics[width=0.35\textwidth]{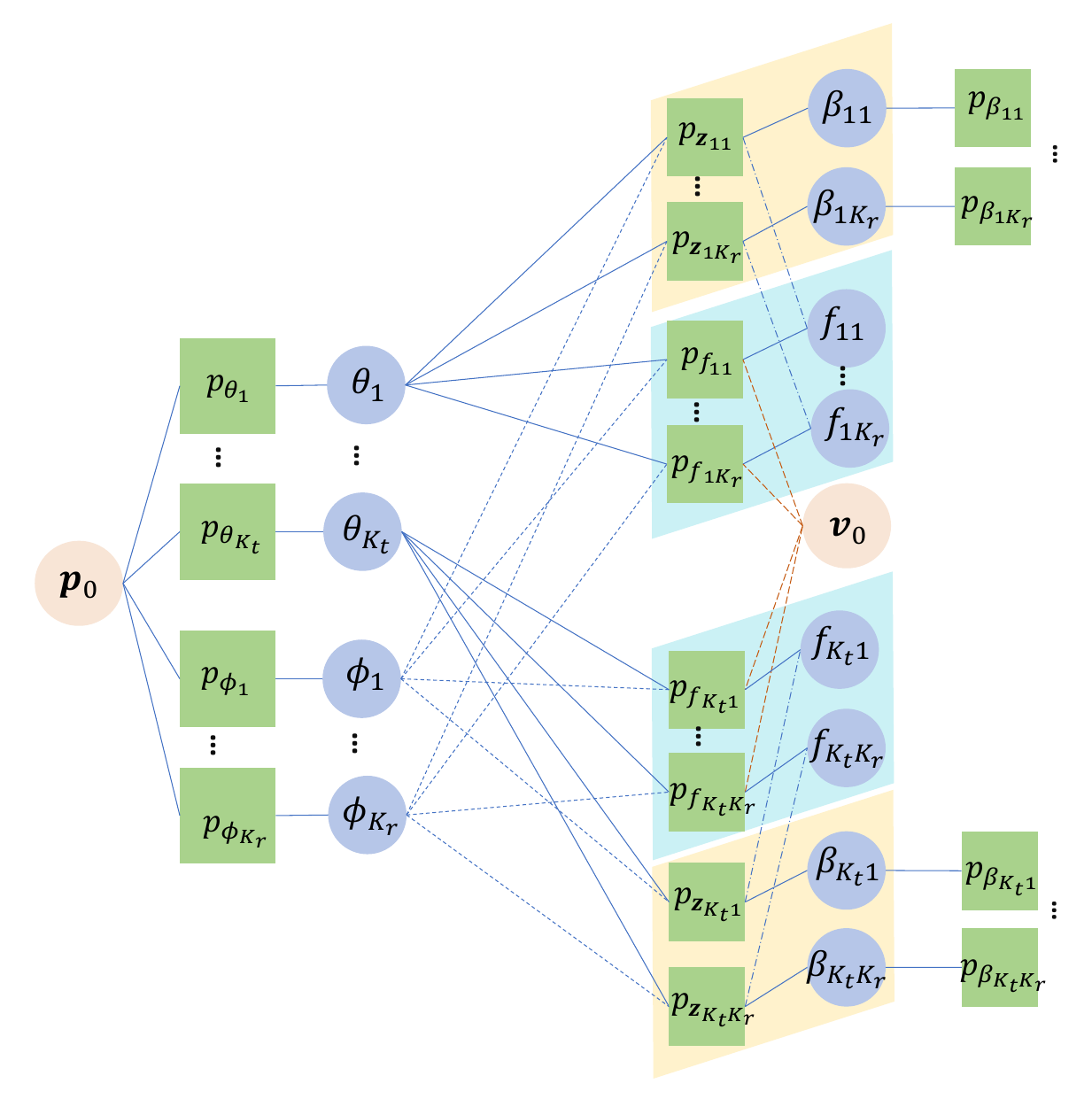}
		\caption{Factor graph representation of \eqref{joint_pdf}, which represents variables as circles and factors as squares, where factors $p({{\mathbf{p}}_{0}})$ and $p({{\mathbf{v}}_{0}})$ are omitted.}
		\label{factor_graph}\vspace{-3mm}
	\end{figure}

\section{Variational Bayesian Inference} \label{section_posteri_pdf_sub}
In this section, we  develop a \ac{vbi} framework   with an alternating optimization algorithm  to compute the approximate posterior distributions of  subarray-level parameters and estimate the noise variance as well as complex amplitude variance hyperparameters.

Taking the $(m,n)$-th \ac{t-r} subarray pair as an example, other subarray pairs can be computed in parallel.
Let $\boldsymbol{\Theta} =\left\{ {{\theta }_{m}},{{\phi }_{n}},{{f}_{mn}},{{\beta }_{mn}} \right\}$ denote the set of parameters to be estimated in the subarray pair, and $\boldsymbol{\alpha }=\left\{ {{\varsigma }_{mn}},\sigma  \right\}$ represent the hyperparameters. 
According to Bayesian rule, we have 
\begin{equation} \label{posteri_pdf_subarray}
   p( \boldsymbol{\Theta}  | {{\mathbf{z}}_{mn}};\boldsymbol{\alpha})\propto \frac{p( {{\mathbf{z}}_{mn}},\boldsymbol{\Theta} ;\boldsymbol{\alpha })}{p ( {{\mathbf{z}}_{mn}};\boldsymbol{\alpha } )},
\end{equation}
where  ${p( {{\mathbf{z}}_{mn}};\boldsymbol{\alpha } )}$ is the model evidence, which acts as a normalizing constant.
The numerator $p( {{\mathbf{z}}_{mn}},\boldsymbol{\Theta} ;\boldsymbol{\alpha })$ is the joint \ac{pdf}, which is given by
\begin{equation} \label{joint_pdf_trsubarray}
        p( {{\mathbf{z}}_{mn}},\Theta;\boldsymbol{\alpha})\!=\!
         p( {{\mathbf{z}}_{mn}} | \Theta ;\boldsymbol{\alpha}) p( {{\theta }_{m}}) p( {{\phi }_{n}}) p( {{f}_{mn}}) p( {{\beta }_{mn}};\boldsymbol{\alpha}).
\end{equation}

Calculating the posterior distributions of subarray-level parameters in $\boldsymbol{\Theta}$ is crucial, as they form the basis for the final location and velocity estimates. 
However, exact Bayesian inference via direct marginalization of the joint posterior in \eqref{posteri_pdf_subarray} becomes computationally intractable due to the high-dimensional integration  and  the non-Gaussian nature of  angular parameters $\theta_m$, $\phi_n$, and Doppler frequency $f_{mn}$,  which precludes closed-form solutions.

Therefore, we resort to the \ac{vbi} methodology\cite{Tzikas2008vbi}, which approximates the true posterior $p(\boldsymbol{\Theta}|\mathbf{z}_{mn};\boldsymbol{\alpha})$ with a factorized surrogate distribution $q(\boldsymbol{\Theta}|\mathbf{z}_{mn})$ while maintaining computational feasibility. 
Specifically, we impose a mean-field approximation that factorizes the joint posterior as
\begin{equation} \label{surrogate_posteri_pdf}
    q( \mathbf{\Theta } | {{\mathbf{z}}_{mn}})\!=\!q( {{\theta }_{m}} | {{\mathbf{z}}_{mn}})q( {{\phi }_{n}} | {{\mathbf{z}}_{mn}} ) q( {{f}_{mn}} | {{\mathbf{z}}_{mn}}) q( {{\beta }_{mn}}| {{\mathbf{z}}_{mn}}).
\end{equation}
The optimal surrogate distribution minimizes the \ac{kld} $D_{KL}(q\|p)$, equivalent to maximizing the \ac{elbo}:
\begin{equation} \label{elbo_origi_express}
    \mathcal{L}\left( q( \mathbf{\Theta }| {{\mathbf{z}}_{mn}} ) \right)={\mathbb{E}_{q( \mathbf{\Theta } | {{\mathbf{z}}_{mn}})}}\left[ \ln \frac{p( {{\mathbf{z}}_{mn}},\Theta ;\boldsymbol{\alpha })}{q( \mathbf{\Theta } | {{\mathbf{z}}_{mn}})} \right].
\end{equation}

Then, we solve the optimization of \eqref{elbo_origi_express} through \ac{cavi}\cite{bishop2006pattern}, iteratively updating each factor $q(u|\mathbf{z}_{mn}), u\in\boldsymbol{\Theta}$ while fixing others. 
For any parameter $u\in\boldsymbol{\Theta}$, the optimal factor satisfies
\begin{equation} \label{trans_elbo_wrt_qu}
    \ln q(u|\mathbf{z}_{mn}) \propto \mathbb{E}_{\backslash u}\left[\ln p(\mathbf{z}_{mn},\boldsymbol{\Theta};\boldsymbol{\alpha})\right],
\end{equation}
where $\propto$ denotes proportionality up to a normalization constant, and ${\mathbb{E}_{\backslash u}}[\cdot]$ denote the expectation over $\frac{q( \mathbf{\Theta} | {{\mathbf{z}}_{mn}} )}{q(u| {{\mathbf{z}}_{mn}} )}$.

To enable unified treatment of spatial-temporal parameters, we introduce the normalized Doppler frequency ${\tilde{f}}_{mn} \triangleq 2\pi T f_{mn}$.
This normalized transformation ensures ${\tilde{f}}_{mn}$ resides on the same $[-\pi,\pi)$ circular domain as $\phi_n$ and $\theta_m$, permitting unified probabilistic modeling through \ac{vm} distributions\footnote{The \ac{vm} distribusion is a circular analogue to Gaussian distributions in linear spaces where it serves as the natural exponential family distribution for periodic variables\cite{bishop2006pattern}.}. 
Following variational line spectral estimation principles \cite{badiu2017variational, zhang2020multiVALSE}, the \ac{vm} parameterization enables closed-form variational updates, while inherently circumventing spectral discretization artifacts caused by grid-based compressed sensing approaches.
The \ac{vm} distribution is formally defined as follows.  
\begin{definition} \label{definition_vmd}
	\emph{ The \ac{vm} distribution, denoted as $\mathcal{M}(\theta ; \mu, \kappa)$, is a continuous probability distribution defined on the circle, often used to model angular variables\cite{bishop2006pattern}. Its probability density function is given by
    \begin{equation}
         \setlength{\abovedisplayskip}{2pt}
	\setlength{\belowdisplayskip}{2pt}
        \mathcal{M}(\theta ; \mu, \kappa)=\frac{1}{2 \pi I_0(\kappa)} \exp (\kappa \cos (\theta-\mu)),
    \end{equation}
    where $\theta \in [-\pi, \pi)$ is the angular variable, $\mu \in [-\pi, \pi)$ is the mean direction, $\kappa \geq 0$ is the concentration parameter that controls the concentration of the distribution around the mean, and $I_0(\cdot)$ is the modified Bessel function of the first kind and order zero. Additionally, the \ac{vm} \ac{pdf} can also be
parameterized in terms of $\eta = \kappa e^{j\mu}$, i.e.,
\begin{equation}
     \setlength{\abovedisplayskip}{2pt}
	\setlength{\belowdisplayskip}{1pt}
    \mathcal{M}(\theta ; \eta)=\frac{1}{2 \pi I_0(|\eta|)} \exp(\Re\{\eta^* e^{j \theta}\}).
\end{equation} }
\end{definition}

\subsection{Inferring the Posterior Probability of \texorpdfstring{$\boldsymbol{\Theta}$}.} 
In this subsection, we derive closed-form variational posterior distributions for the parameter set $\boldsymbol{\Theta}$, alternating between updating $\{\theta_m, \phi_n, \tilde{f}_{mn}\}$  following \ac{vm} distributions and $\beta_{mn}$ as a complex Gaussian. Each parameter is updated while keeping the hyperparameters and other variables fixed, ensuring computationally tractable Bayesian inference.

\begin{lemma}[Variational Posteriors for $\phi_n,\theta_m$, and $\tilde{f}_{mn}$]\label{lemma_subarray_vm}
For each parameter $x \in \{\phi_n,\theta_m,\tilde{f}_{mn}\}$ with VM prior $\mathcal{M}(x;\bar{\eta}_x)$, the optimal posterior approximation $q(x|\mathbf{z}_{mn})$ follows a \ac{vm} distribution  $\mathcal{M}(x;\hat{\eta}_x)$, where:
\begin{subequations}\label{eta_compute}
\begin{align}
\hat{x} &= \bar{x} - {g'(\bar{x})}/{g''(\bar{x})} \\
\hat{\kappa}_x &= A^{-1}\left[\exp\left({1}/{2g''(\bar{x})}\right)\right]
\end{align}
\end{subequations}
with \( g(x) \triangleq \Re[\bar{\eta}_x^*e^{jx} + \sum_{k \in \mathcal{K}_x} \eta_{x,k}^*e^{jkx}] \), where \( \mathcal{K}_x = \{1,...,M-1\} \) for \( x \in \{\phi_n,\theta_m\} \) and \( \mathcal{K}_x = \{1,...,L-1\} \) for \( x = \tilde{f}_{mn} \). 
Additionally, $\bar{x}$ is obtained at the maximum of $g(x)$.
The conjugate parameters are defined as
\begin{subequations}
    \begin{align}
&\tilde{\eta}_{\phi_n,k}^* \!\!\triangleq\!\! \boldsymbol{\eta}_{mn}\!\langle k \rangle\!^H \hat{\mathbf{c}}_{mn}, 
\begin{cases}
\!\boldsymbol{\eta}_{mn} \triangleq -\frac{2}{\sigma}\mathbf{z}_{mn}, \\
\!\hat{\mathbf{c}}_{mn}\! \triangleq \!\!\beta_{mn}\mathbf{a}(\theta_m) \!\!\otimes\! \mathbf{d}(f_{mn})
\end{cases},
\\
&\tilde{\eta}_{\theta_m,k}^*\!\! \triangleq\! \boldsymbol{\eta}_{mn}^{(1)}\!\langle k \rangle\!^H \hat{\mathbf{c}}_{mn}^{(1)}, 
\begin{cases}
\!\boldsymbol{\eta}_{mn}^{(1)} \triangleq -\frac{2}{\sigma}\mathbf{z}_{mn}^{(1)}, \\
\!\hat{\mathbf{c}}_{mn}^{(1)}\!\! \triangleq \!\!\beta_{mn}\mathbf{a}({\phi}_n) \!\!\otimes\! \mathbf{d}({f}_{mn})
\end{cases},
\\
&\eta_{\tilde{f}_{mn},l}^*\!\! \triangleq\! (\hat{\mathbf{c}}_{mn}^{(2)})\!^H\! \boldsymbol{\eta}_{mn}^{(2)}\!\langle l \rangle, 
\begin{cases}
\!\boldsymbol{\eta}_{mn}^{(2)} \triangleq -\frac{2}{\sigma}\mathbf{z}_{mn}^{(2)}, \\
\!{\mathbf{c}}_{mn}^{(2)}\!\! \triangleq\!\! \beta_{mn}\mathbf{a}({\phi}_n)\!\! \otimes \!\mathbf{a}({\theta}_m)
\end{cases},
    \end{align}
\end{subequations}
where \(\mathbf{z}_{mn}^{(i)} = \mathcal{P}_i(\mathbf{z}_{mn})\) denote dimension-rearranged versions of the original measurement vector, and $\langle k\rangle$ denotes the block indexing operator. 
Specifically, we define \(\boldsymbol{\eta}\langle k \rangle \triangleq [\boldsymbol{\eta}]_{kML:(k+1)ML}\), \(\boldsymbol{\eta}^{(1)}\langle k \rangle \triangleq [\boldsymbol{\eta}^{(1)}]_{kML:(k+1)ML}\), and  \(\boldsymbol{\eta}^{(2)}\langle l \rangle \triangleq [\boldsymbol{\eta}^{(2)}]_{lM^2:(l+1)M^2}\).
\end{lemma}
\begin{proof}
Please see Appendix~\ref{proof_of_lemma_subarray_vm}.
\end{proof}

Then, we derive the variational posterior for ${\beta}_{mn}$ by maximizing the \ac{elbo} $\mathcal{L}$ in \eqref{elbo_origi_express} with respect to $q( {{\beta}_{mn}}| {{\mathbf{z}}_{mn}} )$ while keeping other factors fixed. 

\begin{lemma}[Variational  Posterior for ${\beta}_{mn}$] \label{lemma_beta_gaussian}
\emph{Given the complex Gaussian prior $p({\beta}_{mn}) = \mathcal{CN}(0, \varsigma_{mn})$, the optimal posterior approximation $q({\beta}_{mn}|\mathbf{z}_{mn})$ follows a complex Gaussian distribution:
\begin{equation} \label{eq_beta_gaussian_final}
q({\beta}_{mn}|\mathbf{z}_{mn}) = \mathcal{CN}( \hat{\beta}_{mn}, \tilde{\varsigma}_{mn})
\end{equation}
where the posterior mean $\hat{\beta}_{mn}$ and variance $\tilde{\varsigma}_{mn}$ are given by
\begin{equation}\label{eq_beta_mean} 
\begin{aligned}
\hat{\beta}_{mn} = \frac{{{\varsigma }_{mn}}\boldsymbol{\mu }_{mn}^{H}{{\mathbf{z}}_{mn}}}{\sigma +{{\varsigma }_{mn}}\left\| {{\boldsymbol{\mu }}_{mn}} \right\|_{2}^{2}} ,
\tilde{\varsigma}_{mn} = \frac{{{\varsigma }_{mn}}\sigma }{2( \sigma +{{\varsigma }_{mn}}\left\| {{\boldsymbol{\mu }}_{mn}} \right\|_{2}^{2})}
\end{aligned}
\end{equation}
with $\boldsymbol{\mu}_{mn} \triangleq \mathbb{E}_{q}\left[ \mathbf{a}_r({\phi}_n)\otimes \mathbf{a}_t(\theta_m) \otimes \mathbf{d}({{f}_{mn}})\right]$ denoting the expected steering vector, where the expectation is taken over the variational distributions of $\phi_n, \theta_m$, and $f_{mn}$.}
\end{lemma}

\begin{proof}
Please see Appendix~\ref{proof_of_lemma_beta_mn}.
\end{proof}

\begin{remark}
  \emph{Under non-informative priors, \ac{vm} concentration parameters for $\phi_n$, $\theta_m$, and $\tilde{f}_{mn}$ are set close to zero, while complex Gaussian variances $\varsigma_{mn}$ adopt large values. This configuration ensures angular phase wrapping consistency and maintains numerical stability during variational updates.}
\end{remark}

\subsection{Estimating the Hyperparameters \texorpdfstring{$\boldsymbol{\alpha}$}.} 
The noise variance $\sigma$ and the complex amplitude variance  $\varsigma_{mn}$ are typically unknown in practice and must be jointly estimated.
Within the variational Bayesian framework, hyperparameter estimation is performed through a coordinate ascent optimization procedure. 
Specifically, after initializing $\boldsymbol{\alpha}$, the algorithm iteratively performs variational inference to update $q(\boldsymbol{\Theta}|\mathbf{z}_{mn})$ (Section~\ref{section_posteri_pdf_sub}) and computes \ac{mle}-type estimates of $\boldsymbol{\alpha}$ by maximizing the \ac{elbo} with fixed $q(\boldsymbol{\Theta}|\mathbf{z}_{mn})$.

\begin{lemma}[Hyperparameter Estimation] \label{lemma_hyper_est}
The \ac{mle}-type estimates $\sigma$ and $\varsigma_{mn}$ under the variational framework  yield closed-form solutions, which are given by
\begin{subequations}\label{hyper_est_total}
     \setlength{\abovedisplayskip}{2pt}
	\setlength{\belowdisplayskip}{2pt}
\begin{align}
\notag
& \hat{\sigma} =\frac{2}{M^2L}( \|\mathbf{z}_{mn}\|_2^2 + 2\Re(\hat{\beta}_{mn}\mathbf{z}_{mn}^H\boldsymbol{\mu}_{mn}) + \\ 
&\ \ \ \ \ \ \ \ \ \ \ \ \ \|\boldsymbol{\mu}_{mn}\|_2^2(\tilde{\varsigma}_{mn} + |\hat{\beta}_{mn}|^2) ) \label{eq_sigma_est}, \\
& \hat{\varsigma}_{mn} =\tilde{\varsigma}_{mn} + |\hat{\beta}_{mn}|^2, \label{eq_varsigma_est}
\end{align}
\end{subequations}
where $\boldsymbol{\mu}_{mn} \triangleq \mathbb{E}_{q}[\mathbf{a}_r(\phi_n)\otimes\mathbf{a}_t(\theta_m)\otimes\mathbf{d}(f_{mn})]$ represents the expected steering vector, with $\hat{\beta}_{mn}$ and $\tilde{\varsigma}_{mn}$ derived from Lemma~\ref{lemma_beta_gaussian}.
\end{lemma}

\begin{proof}
Please see Appendix~\ref{proof_of_lemma_hyperparameter}.
\end{proof}

\section{Proposed Message Passing-based Estimation Algorithm} \label{section_algorithm}

Building upon the variational posteriors derived in Section~\ref{section_posteri_pdf_sub}, we construct a two-tiered message passing architecture that fuses subarray-level parameter estimates into globally consistent motion parameter inference. 
This framework systematically propagates uncertainty through geometric constraints encoded in the factor graph, where the message operators $\Delta_{a\rightarrow b}(\cdot)$ represent the message  from node $a$ to $b$, while $\Delta_a(\cdot)$ denotes the marginal message at variable node $a$.

\subsection{Subarray Parameter Estimation Module}\label{subsec_angular_doppler}

We present a unified message passing formulation for joint estimation of direction parameters (\ac{dod}s $\{\theta_m\}$, \ac{doa}s $\{\phi_n\}$) and bistatic Doppler frequencies $\{f_{mn}\}$ as follows.
\begin{proposition}[Unified Message Passing Formulation]
\emph{For each \ac{t-r} subarray pair $(m,n)$, the message from variable node $x \in \{\theta_m, \phi_n, f_{mn}\}$ to its connected factor node $p_x$ admits a closed-form \ac{vm} distribution:
\begin{equation} \label{unified_message_passing}
\Delta_{x \rightarrow p_x}(x) \propto \mathcal{M}(x;\mu_{x\rightarrow p_x}, \kappa_{x\rightarrow p_x}),
\end{equation}
where $\mu_{x\rightarrow p_x}$ and $\kappa_{x\rightarrow p_x}$ represent the mean direction and concentration parameter of the distribution, respectively. The explicit derivations of these parameters are presented in the following.}
\end{proposition}

\subsubsection{DoD / DoA Estimation}
According to the sum-product rule \cite{Kschischang2001sum}, the message from variable node $\theta_m$ to factor node $p_{\theta_m}$ is formulated in \eqref{theta_m_to_p_theta_m} at the bottom of this page. Here, $\int_{\backslash{\theta_m}}$ represents the integration over all variables involved in the message computation, excluding ${\theta_m}$.
Due to the high-dimensional integration,  the exact evaluation of $\Delta_{\theta_m \rightarrow p_{\theta_m}}(\theta_m)$ in \eqref{theta_m_to_p_theta_m} is computationally intractable.
	\begin{figure*}[!b]
 	\rule[-0pt]{18.5 cm}{0.05em}
		\begin{equation}
             \setlength{\abovedisplayskip}{2pt}
	\setlength{\belowdisplayskip}{2pt}\label{theta_m_to_p_theta_m}
			\begin{aligned}
				\Delta_{\theta_m \rightarrow p_{\theta_m}}(\theta_m)\propto &\prod\nolimits_{n=1}^{K_r}\! \int_{\backslash \theta_m}\! p(\mathbf{z}_{m n} \mid \phi_n, \theta_m, f_{m n}, \beta_{m n})\! \Delta_{\beta_{m n} \rightarrow p_{z_{m n}}}\!(\beta_{m n})\Delta_{f_{mn} \rightarrow p_{z_{m n}}}(f_{m n}) \Delta_{\phi_n \rightarrow p_{z_{m n}}}(\phi_n) \\
                &\times \prod\nolimits_{n=1}^{K_r} \int_{\backslash{\theta_m}} p(f_{m n} \mid \phi_n, \theta_m, \mathbf{v}_0) \Delta_{f_{m n} \rightarrow p_{f_{m n}}}(f_{m n})\Delta_{\mathbf{v}_0 \rightarrow p_{f_{m n}}}(\mathbf{v}_0) \Delta_{\phi_n \rightarrow p_{f_{m n}}}(\phi_n).
			\end{aligned}
		\end{equation}\vspace{-5mm}
	\end{figure*}

To alleviate this complexity and enable efficient message passing, we resort to the \ac{vm} distribution to approximate the messages related to angular variables such as $\theta_m$.
By treating the message $\Delta_{p_{\theta_m}\rightarrow \theta_m}(\theta_m)$ as the prior distribution of $\theta_m$, the expression of \eqref{theta_m_to_p_theta_m} can be viewed as an estimate of $p(\mathbf{z}_{m n} \mid  \theta_m)$.
Therefore, $\Delta_{\theta_m \rightarrow p_{\theta_m}}(\theta_m)$ can be further expressed as
\begin{equation} \label{thetam_to_pthetam_simp}
     \setlength{\abovedisplayskip}{2pt}
	\setlength{\belowdisplayskip}{2pt}
\begin{aligned}
{{\Delta }_{{{\theta }_{m}}}}_{\to {{p}_{{{\theta }_{m}}}}}( {{\theta }_{m}} )  & \propto \frac{\prod\nolimits_{n=1}^{{{K}_{r}}}{p( {{\mathbf{z}}_{mn}}\left| {{\theta }_{m}} \right. ){{\Delta }_{{{p}_{{{\theta }_{m}}}}\to {{\theta }_{m}}}}( {{\theta }_{m}})}}{{{\Delta }_{{{p}_{{{\theta }_{m}}}}\to {{\theta }_{m}}}}( {{\theta }_{m}} )} \\ 
 & \propto \frac{\prod\nolimits_{n=1}^{{{K}_{r}}}{p( {{\theta }_{m}}\left| {{\mathbf{z}}_{mn}} \right. )}}{{{\Delta }_{{{p}_{{{\theta }_{m}}}}\to {{\theta }_{m}}}}( {{\theta }_{m}} )},  
\end{aligned}
\end{equation}
where the integral is approximated by the posterior distribution $p( {{\theta }_{m}}\left| {{\mathbf{z}}_{mn}} \right. )$, which follows a \ac{vm} distribution, with its explicit form derived in Lemma~\ref{lemma_subarray_vm}. 

The multiplicative closure of \ac{vm} distributions ensures that the product of $p( {{\theta}_{m}}| {{\mathbf{z}}_{mn}} )$ also follows a \ac{vm} distribution, i.e., $\prod_{n=1}^{K_r} p(\theta_m|\mathbf{z}_{mn}) \propto \mathcal{M}(\theta_m; \mu_{\theta_m}, \kappa_{\theta_m})$.    
Given \ac{vm}  prior with known parameters, the forward message ${{\Delta }_{{{p}_{{{\theta }_{m}}}}\to {{\theta }_{m}}}}( {{\theta }_{m}} )$ becomes
\begin{equation}
     \setlength{\abovedisplayskip}{1pt}
	\setlength{\belowdisplayskip}{3pt}
     {{\Delta }_{{{p}_{{{\theta }_{m}}}}\to {{\theta }_{m}}}}( {{\theta }_{m}}t)\propto \mathcal{M}( \pi{\theta }_{m};{{\mu }_{{{p}_{{{\theta }_{m}}}}\to {{\theta }_{m}}}},{\kappa }_{{{p}_{{{\theta }_{m}}}}\to {{\theta }_{m}}} ),
\end{equation}
where for unknown priors (non-informative case), we can simply set ${\kappa }_{{{p}_{{{\theta }_{m}}}}\to {{\theta }_{m}}}\to 0$.

Then, the message ${{\Delta }_{{{\theta }_{m}}}}_{\to {{p}_{{{\theta }_{m}}}}}( {{\theta }_{m}} ) $ can be further approximated by
\begin{equation} \label{VM_theta_m}
     \setlength{\abovedisplayskip}{2pt}
	\setlength{\belowdisplayskip}{2pt}
    \begin{aligned}
   {{\Delta }_{{{\theta }_{m}}\to {{p}_{{{\theta }_{m}}}}}}( {{\theta }_{m}} )&\propto \frac{\mathcal{M}(   {{\theta }_{m}};{{\mu }_{{{\theta }_{m}}}},{\kappa }_{{\theta }_{m}})}{\mathcal{M}(   {{\theta }_{m}};{{\mu }_{{{p}_{{{\theta }_{m}}}}\to {{\theta }_{m}}}},{{\kappa }_{{{p}_{{{\theta }_{m}}}}\to {{\theta }_{m}}}})} \\ 
 & \propto \mathcal{M}(   {{\theta }_{m}};{{\mu }_{{{\theta }_{m}}\to {{p}_{{{\theta }_{m}}}}}},{{\kappa }_{{{\theta }_{m}}\to {{p}_{{{\theta }_{m}}}}}} ), 
\end{aligned}
\end{equation}
where the parameters ${{\mu }_{{{\theta }_{m}}\to {{p}{{{\theta }_{m}}}}}}$ and ${{\kappa }_{{{\theta }_{m}}\to {{p}{{{\theta }_{m}}}}}}$ are obtained by
\begin{equation}
     \setlength{\abovedisplayskip}{3pt}
	\setlength{\belowdisplayskip}{1pt}
  \begin{aligned}
        {{\kappa }_{{{\theta }_{m}}\to {{p}_{{{\theta }_{m}}}}}}\exp ( j{{\mu }_{{{\theta }_{m}}\to {{p}_{{{\theta }_{m}}}}}})=&{\kappa }_{{\theta }_{m}}\exp( j{{\mu }_{{{\theta }_{m}}}} )-\\
        &{{\kappa }_{{{p}_{{{\theta }_{m}}}}\to {{\theta }_{m}}}}\exp ( j{{\mu }_{{{p}_{{{\theta }_{m}}}}\to {{\theta }_{m}}}} ).
  \end{aligned}
\end{equation}

Similarly, the message ${{\Delta }_{{{\phi }_{n}}\to {{p}_{{{\phi }_{n}}}}}}( {{\phi }_{n}} )$ can also be approximated by a \ac{vm} distribution, which is given by
\begin{equation}\label{VM_phi_n}
     \setlength{\abovedisplayskip}{2pt}
	\setlength{\belowdisplayskip}{2pt}
    {{\Delta }_{{{\phi }_{n}}\to {{p}_{{{\phi }_{n}}}}}}( {{\phi }_{n}})\propto \mathcal{M}(   {{\phi }_{n}};{{\mu }_{{{\phi }_{n}}\to {{p}_{{{\phi }_{n}}}}}},{{\kappa }_{{{\phi }_{n}}\to {{p}_{{{\phi }_{n}}}}}} ).
\end{equation}

\subsubsection{Bistatic Doppler Estimation}


Similarly, the message ${{\Delta }_{{{f}_{mn}}\to {{p}_{{{f}_{mn}}}}}}( {{f}_{mn}} )$ can also be approximated by the \ac{vm} distribution, which is given by
\begin{equation}
     \setlength{\abovedisplayskip}{2pt}
	\setlength{\belowdisplayskip}{2pt}
    \begin{aligned}
    {{\Delta }_{{{f}_{mn}}\to {{p}_{{{f}_{mn}}}}}}({{f}_{mn}})&\propto \frac{\mathcal{M}( 2\pi T f_{mn};{{\mu }_{{{f}_{mn}}}},{{\kappa }_{{{f}_{mn}}}})}{\mathcal{M}(2\pi T f_{mn};{{\mu }_{{{p}_{{{f}_{mn}}}}\to {{f}_{mn}}}},{{\kappa }_{{{p}_{{{f}_{mn}}}}\to {{f}_{mn}}}})},\\ 
 & \propto \mathcal{M}( 2\pi T{{f}_{mn}};{{\mu }_{{{f}_{mn}}\to {{p}_{{{f}_{mn}}}}}},{{\kappa }_{{{f}_{mn}}\to {{p}_{{{f}_{mn}}}}}} ),
\end{aligned}
\end{equation}
where ${\mu }_{{f}_{mn}}$ and ${\kappa }_{{f}_{mn}}$ can be calculated by  Lemma~\ref{lemma_subarray_vm}.
Then, the final parameters ${{\mu }_{{{\theta }_{m}}\to {{p}{{{\theta }_{m}}}}}}$ and ${{\kappa }_{{{\theta }_{m}}\to {{p}{{{\theta }_{m}}}}}}$ are obtained by
\begin{equation}
     \setlength{\abovedisplayskip}{2pt}
	\setlength{\belowdisplayskip}{2pt}
  \begin{aligned}
       &{{\kappa }_{{{f}_{mn}}\to {{p}_{{{f}_{mn}}}}}}\!\!\exp(j{{\mu }_{{{f}_{mn}}\to {{p}_{{{f}_{mn}}}}}})\\
       &={{\kappa }_{{{f}_{mn}}}}\!\exp(j{{\mu }_{{{f}_{mn}}}})\!-\!{{\kappa }_{{{p}_{{{f}_{mn}}}}\to {{f}_{mn}}}}\!\exp(j{{\mu }_{{{p}_{{{f}_{mn}}}}\to {{f}_{mn}}}}).
  \end{aligned}
\end{equation}

\subsection{Location Estimation Module}
The spherical wavefront diversity across subarrays creates geometrically distinct observation perspectives, enabling two complementary localization strategies.
The first strategy employs a centralized framework that aggregates all posterior distributions for \ac{doa}s and \ac{dod}s in a central processor for final location estimation. 
The second strategy  adopts a distributed framework where each \ac{t-r} subarray pair independently resolves local position estimates, which are then fused to obtain the final  estimate.


\subsubsection{Centralized  Location Estimation}
The centralized fusion framework propagates directional evidence from all \ac{t-r} subarray pairs through geometric constraints in the factor graph. 
Based on \ac{vm}-form messages $\Delta_{\theta_m\to p_{\theta_m}}$ and $\Delta_{\phi_n\to p_{\phi_n}}$ obtained in \eqref{VM_theta_m} and \eqref{VM_phi_n}, the messages from $p_{\theta_m}$ to $\mathbf{p}_0$ can be given by
\begin{equation}
     \setlength{\abovedisplayskip}{1pt}
	\setlength{\belowdisplayskip}{3pt}
    \begin{aligned}
  & {{\Delta }_{{{p}_{{{\theta }_{m}}}}\to {{\mathbf{p}}_{0}}}}( {{\mathbf{p}}_{0}} )\propto \int_{{{\theta }_{m}}}{p( {{\theta }_{m}}| {{\mathbf{p}}_{0}}){{\Delta }_{{{\theta }_{m}}\to {{p}_{{{\theta }_{m}}}}}}({{\theta }_{m}})}, \\ 
 & \propto \mathcal{M}\Big(  \frac{{{(\mathbf{p}_{m}^{t}-{{\mathbf{p}}_{0}})^{T}}}{{\mathbf{e}}_{y}}}{\|\mathbf{p}_{m}^{t}-{{\mathbf{p}}_{0}}\|};{{\mu }_{{{\theta }_{m}}\to {{p}_{{{\theta }_{m}}}}}},{{\kappa }_{{{\theta }_{m}}\to {{p}_{{{\theta }_{m}}}}}} \Big).  
\end{aligned}
\end{equation}
Similarly, the messages from $p_{{\phi }_{n}}$ to $\mathbf{p}_0$ can be given by
\begin{equation}
     \setlength{\abovedisplayskip}{1pt}
	\setlength{\belowdisplayskip}{3pt}
    \begin{aligned}
  & {{\Delta }_{{{p}_{{{\phi }_{n}}}}\to {{\mathbf{p}}_{0}}}}( {{\mathbf{p}}_{0}})\propto \int_{{{\phi }_{n}}}{p( {{\phi }_{n}} | {{\mathbf{p}}_{0}}){{\Delta }_{{{\phi }_{n}}\to {{p}_{{{\phi }_{n}}}}}}( {{\phi }_{n}})}, \\ 
 & \propto \mathcal{M}\Big(   \frac{{{(\mathbf{p}_{n}^{r}-{{\mathbf{p}}_{0}})^{T}}}{{\mathbf{e}}_{y}}}{\|\mathbf{p}_{n}^{r}-{{\mathbf{p}}_{0}}\|};{{\mu }_{{{\phi }_{n}}\to {{p}_{{{\phi }_{n}}}}}},{{\kappa }_{{{\phi }_{n}}\to {{p}_{{{\phi }_{n}}}}}}\Big). 
\end{aligned}
\end{equation}

Therefore, the marginal probability distribution of $\mathbf{p}_0$ can be computed by the product of all the factor node messages connected to $\mathbf{p}_0$, i.e., 
\begin{equation}
     \setlength{\abovedisplayskip}{4pt}
	\setlength{\belowdisplayskip}{4pt}
    \Delta_{{{\mathbf{p}}_{0}}}({{\mathbf{p}}_{0}})\!\propto\! \prod\nolimits_{m=1}^{{{K}_{t}}}{\prod\nolimits_{n=1}^{{{K}_{r}}}{{{\Delta }_{{{p}_{{{\theta }_{m}}}}\to {{\mathbf{p}}_{0}}}}\!( {{\mathbf{p}}_{0}}){{\Delta }_{{{p}_{{{\phi }_{n}}}}\to {{\mathbf{p}}_{0}}}}\!({{\mathbf{p}}_{0}})}}.
\end{equation}
Based on the central limit theorem, the message $ \mathcal{G}({{\mathbf{p}}_{0}})$ can be approximated as a Gaussian distribution as
\begin{equation} \label{gaussian_approx_p0_pdf}
     \setlength{\abovedisplayskip}{4pt}
	\setlength{\belowdisplayskip}{4pt}
    \Delta_{{{\mathbf{p}}_{0}}}({{\mathbf{p}}_{0}})\propto \mathcal{N}( {{\mathbf{p}}_{0}};{{\mathbf{m}}_{\mathcal{G}}},{{\mathbf{C}}_{\mathcal{G}}} ),
\end{equation}
where ${{\mathbf{m}}_{\mathcal{G}}}$ is the mean vector, and ${{\mathbf{C}}_{\mathcal{G}}}$ is the covariance matrix, which are obtained  in Appendix~\ref{calculate_p0_gaussian_pdf}.

\subsubsection{Distributed Location Estimation} 
For each subarray pair $(m,n)$, the local marginal distribution of $\mathbf{p}_0$ is given by
\begin{equation}\label{product_pthetam_pphin}
     \setlength{\abovedisplayskip}{4pt}
	\setlength{\belowdisplayskip}{4pt}
    \Delta _{{{\mathbf{p}}_{0}}}^{(mn)}({{\mathbf{p}}_{0}})\propto {{\Delta }_{{{p}_{{{\theta }_{m}}}}\to {{\mathbf{p}}_{0}}}}( {{\mathbf{p}}_{0}} ){{\Delta }_{{{p}_{{{\phi }_{n}}}}\to {{\mathbf{p}}_{0}}}}( {{\mathbf{p}}_{0}} ).
\end{equation}
By applying the Laplace approximation for the product of ${{\Delta }_{{{p}_{{{\theta }_{m}}}}\to{{\mathbf{p}}_{0}}}}({{\mathbf{p}}_{0}})$ and ${{\Delta }_{{{p}_{{{\phi }_{n}}}}\to {{\mathbf{p}}_{0}}}}({{\mathbf{p}}_{0}})$, we can obtain the Gaussian approximation of \eqref{product_pthetam_pphin}, i.e.,
\begin{equation}\label{gaussian_approx_dis_p0_pdf}
     \setlength{\abovedisplayskip}{4pt}
	\setlength{\belowdisplayskip}{4pt}
    {\Delta}_{{{\mathbf{p}}_{0}}}^{(mn)}({{\mathbf{p}}_{0}})\propto \mathcal{N}( {{\mathbf{p}}_{0}};\mathbf{m}_{p}^{(mn)},\mathbf{C}_{p}^{(mn)}),
\end{equation}
where $\mathbf{m}_{p}^{(mn)}$ and $\mathbf{C}_{p}^{(mn)}$ are provided in Appendix~\ref{calculate_dis_p0_gaussian_pdf}.

The final posterior distribution of the target location  $\mathbf{p}_0$is obtained by fusing the marginal distributions \eqref{gaussian_approx_dis_p0_pdf} from all \ac{t-r} subarray pairs  through a  product of Gaussians, i.e., $\prod\nolimits_{m=1}^{{{K}_{t}}}{\prod\nolimits_{n=1}^{{{K}_{r}}} {\Delta}_{{{\mathbf{p}}_{0}}}^{(mn)}({{\mathbf{p}}_{0}})}$.
The resulting distribution is a Gaussian with mean vector  and covariance matrix given by
\begin{subequations}\label{p0_sub_level_fusion}
     \setlength{\abovedisplayskip}{4pt}
	\setlength{\belowdisplayskip}{1pt}
    \begin{align}
\mathbf{\tilde{m}}_{\mathcal{G}} &= \mathbf{\tilde{C}}_{\mathcal{G}}\sum\nolimits_{m,n} (\mathbf{C}_p^{(mn)})^{-1}\mathbf{m}_p^{(mn)}, \\
        \mathbf{\tilde{C}}_{\mathcal{G}}^{-1} &= \sum\nolimits_{m,n} (\mathbf{C}_p^{(mn)})^{-1}.
\end{align}
\end{subequations}

\subsection{Velocity Estimation Module}
The velocity estimation process exploits diverse  Doppler frequencies across \ac{t-r} subarray pairs, employing centralized and distributed strategies.

\subsubsection{Centralized Velocity Estimation}  
 For each subarray pair $(m,n)$, the message from the factor node $p_{f_{mn}}$ to the variable node $\mathbf{v}_0$ is given by
\begin{equation}\label{pfmn_to_v0}
     \setlength{\abovedisplayskip}{1pt}
	\setlength{\belowdisplayskip}{4pt}
   \begin{aligned}
       &{{\Delta }_{{{p}_{{{f}_{mn}}}}\to {{\mathbf{v}}_{0}}}}({{\mathbf{v}}_{0}})\!\propto\!\int_{\backslash {{\mathbf{v}}_{0}}}\!p({{f}_{mn}} | {{\phi }_{n}},{{\theta }_{m}},{{\mathbf{v}}_{0}})\times \\
       &{\Delta }_{{\theta }_{m}\to p_{{f}_{mn}}}({{\theta }_{m}}){\Delta }_{{\phi}_{n}\to p_{{f}_{mn}}}( {{\phi }_{n}}) {\Delta }_{f_{mn} \to p_{{f}_{mn}}}( {{f}_{mn}}),
   \end{aligned} 
\end{equation}
where   $\phi_n$ and $\theta_m$ are treated as known parameters and set to their most probable values when estimating the velocity $\mathbf{v}_0$. 
Then, the messages ${\Delta }_{{\theta }_{m}\to p_{{f}_{mn}}}( {{\theta }_{m}})$ and ${\Delta }_{{\phi}_{n}\to p_{{f}_{mn}}}( {{\phi }_{n}} )$  can be ignored as constants when evaluated at their maximum probability values $ {{{\hat{\theta }}}_{m}}={{\mu }_{{{\theta }_{m}}\to {{p}_{{{f}_{mn}}}}}}$ and ${{{\hat{\phi }}}_{n}}={{\mu }_{{{\phi }_{n}}\to {{p}_{{{f}_{mn}}}}}}$. Therefore, the message in \eqref{pfmn_to_v0} can be simplified by
\begin{equation}\label{pfmn_to_v0_simp}
    \begin{aligned}
        & {{\Delta }_{{{p}_{{{f}_{mn}}}}\to {{\mathbf{v}}_{0}}}}({{\mathbf{v}}_{0}})\propto {{\Delta }_{{{f}_{mn}}\to {{p}_{{{f}_{mn}}}}}}( \zeta\mathbf{v}_{0}^{T}{{\mathbf{u}}_{mn}}), 
    \end{aligned}
\end{equation}
where $\zeta\triangleq{2\pi T}/{\lambda }$, and  ${{\mathbf{u}}_{mn}}\triangleq \mathbf{e}(\phi_n)+ \mathbf{e}(\theta_m)$ represents the sum of the unit direction vectors corresponding to the estimated \ac{dod} ${{{\hat{\theta }}}_{m}}$ and \ac{doa} ${{{\hat{\phi }}}_{n}}$.

Based on \eqref{pfmn_to_v0_simp}, the marginal probability distribution of ${\mathbf{v}}_0$ can be obtained by
\begin{equation} \label{v0_marg_prob}
     \setlength{\abovedisplayskip}{2pt}
	\setlength{\belowdisplayskip}{4pt}
    {{\Delta }_{{{\mathbf{v}}_{0}}}}({{\mathbf{v}}_{0}})\propto \prod\nolimits_{m=1}^{{{K}_{t}}}{\prod\nolimits_{n=1}^{{{K}_{r}}}{{{\Delta }_{{{f}_{mn}}\to {{p}_{{{f}_{mn}}}}}}(\zeta \mathbf{v}_{0}^{T}{{\mathbf{u}}_{mn}})}}.
\end{equation}
By applying the Laplace approximation, we can obtain a Gaussian approximation of ${{\Delta }_{{\mathbf{v}}_0}}( {{\mathbf{v}}_{0}})$, which is given by
\begin{equation}\label{gaussian_approx_v0_pdf}
     \setlength{\abovedisplayskip}{3pt}
	\setlength{\belowdisplayskip}{3pt}
    {{\Delta }_{{{\mathbf{v}}_{0}}}}({{\mathbf{v}}_{0}} )\propto \mathcal{C}\mathcal{N}({{\mathbf{v}}_{0}};{{\mathbf{m}}_{\mathcal{H}}},{{\mathbf{C}}_{\mathcal{H}}}),
\end{equation}
where the mean vector ${\mathbf{m}}_{\mathcal{H}}$ and the covariance matrix  ${\mathbf{C}}_{\mathcal{H}}$ are given in Appendix~\ref{calculate_v0_gaussian_pdf}.

\subsubsection{Distributed Velocity Estimation}
In contrast to location estimation,  the ambiguity in the velocity vector direction  arises when using a single subarray pair. 
To resolve this problem, we employ geometric diversity through distinct \ac{t-r} subarray configurations.
Let $\Omega$ denote the set of all possible configurations of two distinct \ac{t-r} subarray pairs, with each configuration represented by $\omega_i = \{(m,n), (p,q)\}$, where $i = 1, 2, \ldots, |\Omega|$, and $|\Omega|$ represents the total number of configurations. To address the velocity estimation issue, we select a subset of configurations $\omega_i \in \Omega$, such that $m \neq p$ or $n \neq q$, to obtain the marginal distributions of $\mathbf{v}_0$.
The marginal distribution of $\mathbf{v}_0$ determined by the configuration $\omega_i$ is given by
\begin{equation} \label{vo_marg_pdf_dis_pair}
     \setlength{\abovedisplayskip}{1pt}
	\setlength{\belowdisplayskip}{5pt}
    \Delta_{{{\mathbf{v}}_{0}}}^{{({\omega }_{i})}}\!( {{\mathbf{v}}_{0}})\!\propto\! {{\Delta }_{{{f}_{mn}}\!\to{{p}_{{{f}_{mn}}}}}}\!(\zeta \mathbf{v}_{0}^{T}\!{{\mathbf{u}}_{mn}}){{\Delta }_{{{f}_{pq}}\!\to {{p}_{{{f}_{pq}}}}}}\!(\zeta \mathbf{v}_{0}^{T}\!{{\mathbf{u}}_{pq}}).
\end{equation}
By applying the Laplace approximation for the  product of $ {{\Delta }_{{{f}_{mn}}\to {{p}_{{{f}_{mn}}}}}}(\zeta \mathbf{v}_{0}^{T}{{\mathbf{u}}_{mn}})$ and ${{\Delta }_{{{f}_{pq}}\to {{p}_{{{f}_{pq}}}}}}(\zeta \mathbf{v}_{0}^{T}{{\mathbf{u}}_{pq}})$, we can obtain the Gaussian approximation of $ \Delta _{{{\mathbf{v}}_{0}}}^{{({\omega }_{i})}}({{\mathbf{v}}_{0}})$, i.e.,
\begin{equation}\label{gaussian_approx_v0_dis_pdf}
{\Delta}_{{{\mathbf{v}}_{0}}}^{({\omega }_{i})}({{\mathbf{p}}_{0}})\propto \mathcal{N}( {{\mathbf{p}}_{0}};\mathbf{m}_{v}^{({\omega }_{i})},\mathbf{C}_{v}^{({\omega }_{i})}),
\end{equation}
where $\mathbf{m}_{v}^{({\omega }_{i})}$ and $\mathbf{C}_{v}^{({\omega }_{i})}$ are given in Appendix~\ref{calculate_dis_v0_gaussian_pdf}.

Subsequently, distributions in \eqref{gaussian_approx_v0_dis_pdf} are fused by a product of Gaussians to acquire the final velocity estimate, where the mean vector and covariance matrix given by:
\begin{subequations}\label{v0_sub_level_fusion}
       \setlength{\abovedisplayskip}{2pt}
	\setlength{\belowdisplayskip}{2pt}
    \begin{align}
& {{{\mathbf{\tilde{m}}}}_{\mathcal{H}}}={{{\mathbf{\tilde{C}}}}_{\mathcal{H}}}(\sum\nolimits_{i=1}^{|\Omega |}{{{( \mathbf{C}_{v}^{({{\omega }_{i}})} )}^{-1}}}\mathbf{m}_{v}^{({{\omega }_{i}})}), \\ 
 & \mathbf{\tilde{C}}_{\mathcal{H}}^{-1}=\sum\nolimits_{i=1}^{|\Omega |}{{{(\mathbf{C}_{v}^{({{\omega }_{i}})})}^{-1}}}. 
\end{align}
\end{subequations}


%

%

\subsection{Overall Algorithm}

Based on previous sections, we propose the Variational Inference and Message Passing-based Near-Field Motion Parameter Estimation (VMP-NMPE) algorithm, as summarized in Algorithm~\ref{alg:vmp_nmpe}. 
The proposed framework operates through two cascaded stages:

\textbf{Stage 1 - Variational Bayesian Inference:} Each \ac{t-r} subarray pair executes parallel computation of closed-form posterior distributions for intermediate parameters (\ac{doa} $\phi_n$, \ac{dod} $\theta_m$, bistatic Doppler $f_{mn}$, reflection coefficient $\beta_{mn}$) and hyperparameters ($\sigma$, $\varsigma_{mn}$) through \ac{cavi}, as summarized in Algorithm~\ref{alg:vmp_subarray}.

\textbf{Stage 2 - Hierarchical Message Passing:} A two-tiered message passing architecture, which contains two modes:
\begin{itemize}
    \item \textit{System-level fusion}: Aggregates all subarray messages through gradient-based optimization for centralized location and velocity estimation.
    \item \textit{Subarray-level fusion}: Performs distributed estimation, fusing results via Gaussian mixture product rules to obtain final location and velocity estimates.
\end{itemize}

The computational complexity of the VMP-NMPE algorithm is as follows:
The variational inference stage requires $O(K_tK_rN_{vi}(M^2 + ML))$ operations for $N_{vi}$ iterations across $K_tK_r$ subarrays. The location estimation via system-level fusion involves $O(N_{gd}K_tK_r)$ operations for gradient descent with $N_{gd}$ steps, while the velocity estimation entails $O(N_{ga}K_tK_r)$ operations for Gauss-Newton iterations with $N_{ga}$ steps. 
The subarray-level distributed fusion incurs $O(|\Omega|)$ operations for $|\Omega|$ subarray pair configurations. 


\begin{algorithm}[t]
\caption{Variational Bayesian Inference for Subarray Pair}
\label{alg:vmp_subarray}
\begin{algorithmic}[1]
\REQUIRE $\mathbf{z}_{mn}$, $\bar{\eta}_{\theta_m}$, $\bar{\eta}_{\phi_n}$, $\bar{\eta}_{f_{mn}}$, $\varsigma_{mn}^{(0)}$, $\hat{\beta}_{mn}^{(0)}$, $\sigma^{(0)}$,  $N_{\text{max}}$, $\epsilon$
\ENSURE Posterior parameters $\hat{\eta}_{\theta_m}$, $\hat{\eta}_{\phi_n}$, $\hat{\eta}_{f_{mn}}$, $\hat{\beta}_{mn}$, and hyperparameters $\tilde{\varsigma}_{mn}$, $\hat{\sigma}$
\STATE Initialize $q^{(0)}(\theta_m) \gets \mathcal{M}(\theta_m;\bar{\eta}_{\theta_m})$,  $q^{(0)}(\phi_n) \gets \mathcal{M}(\phi_n;\bar{\eta}_{\phi_n})$, $q^{(0)}(f_{mn}) \gets \mathcal{M}(f_{mn};\bar{\eta}_{f_{mn}})$, $q^{(0)}(\beta_{mn}) \gets \mathcal{CN}(0, \varsigma_{mn}^{(0)})$, and $\boldsymbol{\xi}^{(0)} \gets [{\bar{\eta}_{\theta_m}, \bar{\eta}_{\phi_n}, \bar{\eta}_{f_{mn}}, \hat{\beta}_{mn}^{(0)}, \tilde{\varsigma}_{mn}^{(0)}, \hat{\sigma}^{(0)}}]$
\FOR{$i = 1$ to $N_{\text{max}}$}
    \STATE Update angular and posteriors
     $\hat{\eta}_{\phi_n}^{(i)} \gets \text{Solve } \eqref{eta_compute}$ with $x = \phi_n$, 
    $\hat{\eta}_{\theta_m}^{(i)} \gets \text{Solve } \eqref{eta_compute}$ with $x = \theta_m$
    \STATE Update Doppler posterior:
     $\hat{\eta}_{f_{mn}}^{(i)} \gets \text{Solve } \eqref{eta_compute}$ with $x = \tilde{f}_{mn}$
    \STATE Update reflection coefficient:
     $\hat{\beta}_{mn}^{(i)}, \tilde{\varsigma}_{mn}^{(i)} \gets \eqref{eq_beta_mean}$
    \STATE Estimate hyperparameters:
     $\hat{\sigma}^{(i)} \gets \eqref{eq_sigma_est}$,
     $\hat{\varsigma}_{mn}^{(i)} \gets \eqref{eq_varsigma_est}$
     \STATE Update $\boldsymbol{\xi}^{(i)} \gets [{\hat{\eta}_{\theta_m}^{(i)}, \hat{\eta}_{\phi_n}^{(i)}, \hat{\eta}_{f_{mn}}^{(i)}, \hat{\beta}_{mn}^{(i)}, \tilde{\varsigma}_{mn}^{(i)}, \hat{\sigma}^{(i)}}]$
    \IF{$\|\boldsymbol{\xi}^{(i)} - \boldsymbol{\xi}^{(i-1)}\| < \epsilon$}
        \STATE \textbf{break}
    \ENDIF
\ENDFOR
\end{algorithmic}
\end{algorithm}

\begin{algorithm}[t]
\caption{VMP-NMPE: Variational Message Passing for Near-Field Motion Parameter Estimation}
\label{alg:vmp_nmpe}
\begin{algorithmic}[1]
\REQUIRE $\{\mathbf{z}_{mn}\}$, system parameters $T, \lambda$, convergence thresholds $\epsilon, \epsilon_1, \epsilon_2$
\ENSURE Target location $\hat{\mathbf{p}}_0$, velocity $\hat{\mathbf{v}}_0$, covariance matrices $\mathbf{C}_{\mathcal{G}}, \mathbf{C}_{\mathcal{H}}$

\STATE \textbf{Stage 1: Distributed Variational Inference}
\FOR{each transmit subarray $m \in [1,K_t]$}
    \FOR{each receive subarray $n \in [1,K_r]$}
        \STATE Run Algorithm \ref{alg:vmp_subarray} for subarray pair $(m,n)$
        \STATE Store posterior parameters $\hat{\eta}_{\theta_m}, \hat{\eta}_{\phi_n}, \hat{\eta}_{f_{mn}}, \hat{\beta}_{mn}$
\STATE Compute local messages $\Delta_{p_{\theta_m} \rightarrow\theta_m}$, $\Delta_{p_{\phi_n} \rightarrow \phi_n}$, $\Delta_{p_{f_{mn}} \rightarrow f_{mn}}$ using \eqref{unified_message_passing}
    \ENDFOR
\ENDFOR

\STATE \textbf{Stage 2: Hierarchical Message Passing}
\STATE \textbf{Option A: System-level Fusion}
\STATE Initialize $\mathbf{p}_0^{(0)}$, set iteration $i=0$
\WHILE{$\|\mathbf{p}_0^{(i+1)} - \mathbf{p}_0^{(i)}\| > \epsilon_1$}
\STATE Update $\mathbf{p}_0^{(i+1)}$ using \eqref{p0_sys_level_update}
\STATE $i \leftarrow i+1$
\ENDWHILE
\STATE Set $\hat{\mathbf{p}}_0 = \mathbf{p}_0^{(i)}$, compute $\mathbf{C}_{\mathcal{G}}$ using \eqref{covar_p0_sys_level}
\STATE Initialize $\mathbf{v}_0^{(0)}$, set iteration $j=0$
\WHILE{$\|\mathbf{v}_0^{(j+1)} - \mathbf{v}_0^{(j)}\| > \epsilon_2$}
\STATE Update $\mathbf{v}_0^{(j+1)}$ using \eqref{v0_sys_update}
\STATE $j \leftarrow j+1$
\ENDWHILE
\STATE Set $\hat{\mathbf{v}}_0 = \mathbf{v}_0^{(j)}$, compute $\mathbf{C}_{\mathcal{H}}$ using \eqref{v0_sys_level_covar}
\STATE \textbf{Option B: Subarray-level Fusion}
\FOR{each subarray pair $(m,n)$}
\STATE Compute local estimates $\mathbf{m}_p^{(mn)}, \mathbf{C}_p^{(mn)}$ using \eqref{gaussian_approx_dis_p0_pdf}
\ENDFOR
\STATE Fuse local estimates ${\mathbf{m}_p^{(mn)}, \mathbf{C}_p^{(mn)}}$ to obtain $\hat{\mathbf{p}}_0, \mathbf{C}_{\mathcal{G}}$ using \eqref{p0_sub_level_fusion}
\FOR{each subarray pair configuration $\omega_i \in \Omega$}
\STATE Compute local estimates $\mathbf{m}_v^{(\omega_i)}, \mathbf{C}_v^{(\omega_i)}$ using \eqref{gaussian_approx_v0_dis_pdf}
\ENDFOR
\STATE Fuse local estimates ${\mathbf{m}_v^{(\omega_i)}, \mathbf{C}_v^{(\omega_i)}}$ to obtain $\hat{\mathbf{v}}_0, \mathbf{C}_{\mathcal{H}}$ using \eqref{v0_sub_level_fusion}

\end{algorithmic}
\end{algorithm}

\section{Cramer-Rao Bound} \label{section_performance_analysis}
To evaluate the near-field motion parameter estimation, we derive the \ac{crb}s for estimating $\boldsymbol{\eta }={{[\mathbf{p}_{0}^{T},\mathbf{v}_{0}^{T},{\tilde{\boldsymbol{\beta}}^{T}}]}^{T}}$.
We first define an intermediate parameter set $\boldsymbol{\rho}=[{{\boldsymbol{\theta }}^{T}},{\boldsymbol{\phi }^{T}},{{\boldsymbol{f}}^{T}},{\tilde{\boldsymbol{\beta }}^{T}}]^{T}$, where $\boldsymbol{\theta }\triangleq[ {{\theta }_{1}},\ldots ,{{\theta }_{{{K}_{t}}}}]$, $\boldsymbol{\phi } \triangleq[ {{\phi }_{1}},\ldots ,{{\phi }_{{{K}_{r}}}}]$, $\boldsymbol{f}\triangleq[ {{f}_{11}},{{f}_{12}},\ldots ,{{f}_{{{K}_{t}}{{K}_{r}}}}]$, and $\tilde{\boldsymbol{\beta}}\triangleq[ \Re\{ {{\beta }_{11}}\},\ldots ,\Re \{ {{\beta }_{{{K}_{t}}{{K}_{r}}}} \},\Im \{ {{\beta }_{11}} \},\ldots ,\Im \{ {{\beta }_{{{K}_{t}}{{K}_{r}}}} \} ] $.

The \ac{crb} for $\boldsymbol{\eta}$ satisfies
\begin{equation}
       \setlength{\abovedisplayskip}{2pt}
	\setlength{\belowdisplayskip}{2pt}
    \mathbb{E}\{ (\boldsymbol{\hat{\eta }}-\boldsymbol{\eta }){( \boldsymbol{\hat{\eta }}-\boldsymbol{\eta })^{T}} \}\succeq {{\mathbf{F}}_{\boldsymbol{\eta }}^{-1}},
\end{equation}
where ${\mathbf{F}}_{\boldsymbol{\eta }}$ denotes the \ac{fim} for $\boldsymbol{\eta}$, which is given by
\begin{equation}
       \setlength{\abovedisplayskip}{3pt}
	\setlength{\belowdisplayskip}{3pt}
    \mathbf{F}_{\boldsymbol{\eta}}\!\!=\!\!\frac{2}{\sigma}\!\! \sum\limits_{m=1}^{{{K}_{t}}}\!{\sum\limits_{n=1}^{{{K}_{r}}}\!{\sum\limits_{l=0}^{L-1}\!{\Re \!\bigg\{ \!\!\Big(\!\frac{\partial ( {{\beta }_{mn}}\boldsymbol{\mu }_{mn}^{(l)})}{\partial \boldsymbol{\eta }}\!\Big)^H\!\!{{\Big(\frac{\partial ( {{\beta }_{mn}}\boldsymbol{\mu }_{mn}^{(l)} )}{\partial \boldsymbol{\eta }}\!\Big)}}\!\!\bigg\}}}}.
\end{equation}
However, direct computation of $\mathbf{F}$ is challenging due to the nonlinear coupling between $\boldsymbol{\eta}$ and the physical observation model. To address this, we first derive the \ac{fim} for the intermediate parameters $\boldsymbol{\rho}$, followed by a transformation via the chain rule.

The \ac{fim} for $\boldsymbol{\rho}$ is formulated as:
\begin{equation}\label{fim_calcu_rho}
       \setlength{\abovedisplayskip}{3pt}
	\setlength{\belowdisplayskip}{3pt}
\mathbf{F}_{\boldsymbol{\rho}}\!\!=\!\!\frac{2}{\sigma }\!\!\sum\limits_{m=1}^{{{K}_{t}}}\!{\sum\limits_{n=1}^{{{K}_{r}}}\!{\sum\limits_{l=0}^{L-1}\!{\Re\!\bigg\{\!\!{\Big(\!\frac{\partial({{\beta }_{mn}}\boldsymbol{\mu }_{mn}^{(l)})}{\partial \boldsymbol{\rho}}\!\!\Big)^H}\!\!\Big(\!\frac{\partial ( {{\beta }_{mn}}\boldsymbol{\mu }_{mn}^{(l)})}{\partial \boldsymbol{\rho}}\!\!\Big)\!\!\bigg\}}}}.
\end{equation}
Then, by using the Jacobian matrix $\boldsymbol{\Psi} \triangleq \frac{\partial \boldsymbol{\rho}}{\partial \boldsymbol{\eta }}$ and the chain rule $\frac{\partial( {{\beta }_{mn}}\boldsymbol{\mu }_{mn}^{(l)} )}{\partial \boldsymbol{\eta }} = \frac{\partial( {{\beta }_{mn}}\boldsymbol{\mu }_{mn}^{(l)} )}{\partial \boldsymbol{\rho}} \boldsymbol{\Psi }$, the \ac{fim} for $\boldsymbol{\eta}$ can be obtained through the transformation:
\begin{equation}\label{fim_calcu_eta}
       \setlength{\abovedisplayskip}{2pt}
	\setlength{\belowdisplayskip}{1pt}
\mathbf{F}_{\boldsymbol{\eta}} = \boldsymbol{\Psi }^{T} \mathbf{F}_{\boldsymbol{\rho}} \boldsymbol{\Psi }.
\end{equation}

The calculations for the terms in $\frac{\partial( {{\beta }_{mn}}\boldsymbol{\mu }_{mn}^{(l)})}{\partial \boldsymbol{\rho}}$ and $\boldsymbol{\Psi}$ are derived as following.
The nonzero entries of $\frac{\partial( {{\beta }_{mn}}\boldsymbol{\mu }_{mn}^{(l)})}{\partial \boldsymbol{\rho}}$ are calculated by
\begin{equation}
     \setlength{\abovedisplayskip}{2pt}
	\setlength{\belowdisplayskip}{2pt}
    \frac{\partial( {{\beta }_{mn}}\boldsymbol{\mu }_{mn}^{(l)})}{\partial \Re\{ {{\beta }_{mn}} \}}=\boldsymbol{\mu }_{mn}^{(l)},\quad \frac{\partial( {{\beta }_{mn}}\boldsymbol{\mu }_{mn}^{(l)} )}{\partial \Im \{ {{\beta }_{mn}} \}}=j\boldsymbol{\mu }_{mn}^{(l)},
\end{equation}
\begin{equation}
    \begin{aligned}
  & \frac{\partial ({{\beta }_{mn}}\boldsymbol{\mu }_{mn}^{(l)})}{\partial {{\theta }_{m}}}\!\!=\!j{{\beta }_{mn}}[ {{e}^{-j2\pi {{f}_{mn}}{{t}_{l}}}}{{\mathbf{a}}_{r}}({{\phi }_{n}})\!\otimes\!( \mathbf{c}\!\odot\! {{\mathbf{a}}_{t}}({{\theta }_{m}}))\\ 
 & \!+{2\pi {{t}_{l}}}/{\lambda }{{e}^{-j2\pi {{f}_{mn}}{{t}_{l}}}}(\tan({{{\tilde{\theta }}}_{m}}){{v}_{x}}\!-\!{{v}_{y}}){{\mathbf{a}}_{r}}({{\phi }_{n}})\!\otimes\! {{\mathbf{a}}_{t}}({{\theta}_{m}})], 
\end{aligned}
\end{equation}
where $\mathbf{c}\triangleq{{\left[ 0,1,\ldots ,M-1 \right]}^{T}}$. Similarly, we have
\begin{equation}
       \setlength{\abovedisplayskip}{2pt}
	\setlength{\belowdisplayskip}{2pt}
    \begin{aligned}
  & \frac{\partial({{\beta }_{mn}}\boldsymbol{\mu }_{mn}^{(l)})}{\partial {{\phi }_{n}}}\!=\!j{{\beta }_{mn}}[ {{e}^{-j2\pi {{f}_{mn}}{{t}_{l}}}}( \mathbf{c}\!\odot \!{{\mathbf{a}}_{r}}({{\phi }_{n}}))\!\otimes \!{{\mathbf{a}}_{t}}({{\theta }_{m}}) \\ 
 & \!+{2\pi {{t}_{l}}}/{\lambda }{{e}^{-j2\pi {{f}_{mn}}{{t}_{l}}}}(\tan({{{\tilde{\phi}}}_{n}} ){{v}_{x}}\!-\!{{v}_{y}}){{\mathbf{a}}_{r}}({{\phi }_{n}})\!\otimes\! {{\mathbf{a}}_{t}}({{\theta }_{m}})],
\end{aligned}
\end{equation}
\begin{equation}
       \setlength{\abovedisplayskip}{2pt}
	\setlength{\belowdisplayskip}{2pt}
    \begin{aligned}
  & \frac{\partial({{\beta }_{mn}}\boldsymbol{\mu }_{mn}^{(l)})}{\partial {{f}_{mn}}}\!\!=\!\!j{{\beta }_{mn}}{{e}^{\!-j2\pi {{f}_{mn}}{{t}_{l}}}}\!\bigg[ \!\frac{\lambda( \mathbf{c}\!\odot\! {{\mathbf{a}}_{r}}({{\phi }_{n}}))}{{{v}_{y}}\!-\!\tan({{{\tilde{\phi }}}_{n}}){{v}_{x}}}\!\!\otimes\! {{\mathbf{a}}_{t}}({{\theta }_{m}})  \\ 
 & \!-2\pi {{t}_{l}}{{\mathbf{a}}_{r}}({{\phi }_{n}})\!\otimes\!{{\mathbf{a}}_{t}}({{\theta }_{m}})\!+\!{{\mathbf{a}}_{r}}({{\phi }_{n}})\!\otimes\! \frac{\lambda( \mathbf{c}\!\odot\! {{\mathbf{a}}_{t}}({{\theta }_{m}}))}{{{v}_{y}}\!-\!\tan( {{{\tilde{\theta }}}_{m}}){{v}_{x}}} \!\bigg].
\end{aligned}
\end{equation}

The nonzero entries in $\boldsymbol{\Psi}$ are calculated by
\begin{equation}
       \setlength{\abovedisplayskip}{2pt}
	\setlength{\belowdisplayskip}{2pt}
    \frac{\partial {{\theta }_{m}}}{\partial {{\mathbf{p}}_{0}}}=\frac{\partial {( \mathbf{e}_{m}^{t} )^{T}}{{\mathbf{e}}_{x}}}{\partial {{\mathbf{p}}_{0}}}=\frac{-{{\mathbf{e}}_{x}}+{( \mathbf{e}_{m}^{t})^{T}}{{\mathbf{e}}_{x}}{( \mathbf{e}_{m}^{t})^{T}}}{{{\| \mathbf{p}_{m}^{t}-{{\mathbf{p}}_{0}}\|}_{2}}},
\end{equation}
\begin{equation}
    \frac{\partial {{\phi }_{n}}}{\partial {{\mathbf{p}}_{0}}}=\frac{\partial {( \mathbf{e}_{n}^{r} )^{T}}{{\mathbf{e}}_{x}}}{\partial {{\mathbf{p}}_{0}}}=\frac{-{{\mathbf{e}}_{x}}+{( \mathbf{e}_{n}^{r})^{T}}{{\mathbf{e}}_{x}}{( \mathbf{e}_{n}^{r})^{T}}}{{{\| \mathbf{p}_{n}^{r}-{{\mathbf{p}}_{0}}\|}_{2}}},
\end{equation}
\begin{equation}
   \begin{aligned}
       \frac{\partial {{f}_{mn}}}{\partial {{\mathbf{p}}_{0}}} \!\!=\!\!\frac{-{{\mathbf{v}}_{0}}\!\!+\!\!{{(\mathbf{e}_{m}^{t})^{T}}}\!{{\mathbf{v}}_{0}}{{(\mathbf{e}_{m}^{t})^{T}}}}{{{\| \mathbf{p}_{m}^{t}-{{\mathbf{p}}_{0}}\|}_{2}}}\!\!+\!\!\frac{-\!{{\mathbf{v}}_{0}}\!\!+\!\!{{( \mathbf{e}_{n}^{r} )^{T}}}\!{{\mathbf{v}}_{0}}{{( \mathbf{e}_{n}^{r})^{T}}}}{{{\| \mathbf{p}_{n}^{r}-{{\mathbf{p}}_{0}} \|}_{2}}},  
   \end{aligned}
\end{equation}
\begin{equation}
    \frac{\partial {{f}_{mn}}}{\partial {{\mathbf{v}}_{0}}}=\mathbf{e}_{m}^{t}+\mathbf{e}_{n}^{r},\ \frac{\partial \Re \{ {{\beta }_{mn}}\}}{\partial \Re \{ {{\beta }_{mn}}\}}=1,\ \frac{\partial \Im\{ {{\beta }_{mn}} \}}{\partial \Im \{ {{\beta }_{mn}}\}}=1.
\end{equation}

\section{Numerical Simulation} \label{section_simulation}
In this section, we presents a comprehensive performance evaluation of the proposed VMP-NMPE algorithm through numerical simulations.
We compare our method against three benchmarks: 
\begin{enumerate}
\item \textbf{Maximum Likelihood Estimation} (labeled ``ML Estimation"): This scheme jointly estimates the target location $\mathbf{p}_0$, velocity $\mathbf{v}_0$, and complex amplitudes $\boldsymbol{\beta}$ by solving the optimization problem\cite{wang2024velocity}:
\begin{equation}
   \begin{aligned}
        \underset{{{\mathbf{p}}_{0}},{{\mathbf{v}}_{0}}}{\mathop{\min }}\,\!\sum\nolimits_{m=1}^{{{K}_{t}}}\!{\sum\nolimits_{n=1}^{{{K}_{r}}}\!\!{\| {{\mathbf{z}}_{mn}}\!\!-\!{{\beta }_{mn}}{{\boldsymbol{\mu }}_{mn}}( {{\mathbf{p}}_{0}},{{\mathbf{v}}_{0}} )\|_{2}^{2}}}
   \end{aligned}
\end{equation}
The gradient descent algorithm is employed to solve this problem and obtain the ML estimates of $\mathbf{p}_0$ and $\mathbf{v}_0$.

\item \textbf{Grid-based Method}: This is a two-stage algorithm that first estimates the location using a grid search approach. Then, it employs the 2D \ac{music} algorithm to estimate the velocity based on the estimated location.

\item \textbf{Subarray-level Average}: In this scheme, each \ac{t-r} subarray pair independently estimates the location and velocity. These subarray-level estimates are then averaged to obtain the final location and velocity estimates. 
\end{enumerate}  
The system configuration parameters are detailed in Table \ref{simulation_param}. 
We quantify the estimation performance using the \ac{rmse} metrics for both location and velocity:
\begin{subequations}
    \begin{align}
        & \text{RMSE}({{\mathbf{p}}_{0}})\triangleq \sqrt{\mathbb{E}( {{| {{{\hat{x}}}_{0}}-{{x}_{0}} |^{2}}}+{{| {{{\hat{y}}}_{0}}-{{y}_{0}} |^{2}}})}, \\ 
 & \text{RMSE}( {{\mathbf{v}}_{0}})\triangleq \sqrt{\mathbb{E}({{|{{{\hat{v}}}_{x}}-{{v}_{x}}|^{2}}}+{{| {{{\hat{v}}}_{y}}-{{v}_{y}} |^{2}}})}. 
    \end{align}
\end{subequations}
The derived \ac{crb}s are served as theoretical performance benchmarks for both location and velocity estimation. 
All numerical results are statistically validated through 200 independent Monte Carlo trials.

\begin{table} 
\centering
\caption{Radar Settings Used in the Simulations}\label{simulation_param}
\begin{tabular}{|c|c|c|c|}
\hline
Parameter & Value & Parameter & Value \\
\hline
$P_T$ & 30 dBm &  $\sigma$ &  {1 $\text{m}^2$} \\
$T$ & 10 $\mu$s &  $L$ & 600 \\
$B$ & 200 MHz &  $f_c$ & 28 GHz\\
$G_T / G_R$ & 15 dB & $N_t / N_r$ & 256  \\
 $\mathbf{p}_0$ & (15,20.7) m & $\mathbf{v}_0$ & (10,10.2) m/s  \\
\hline
\end{tabular}\vspace{-3mm}
\end{table}

\begin{figure} [t]
\setlength{\abovecaptionskip}{-0.1cm}
		\centering		\includegraphics[width=0.3\textwidth]{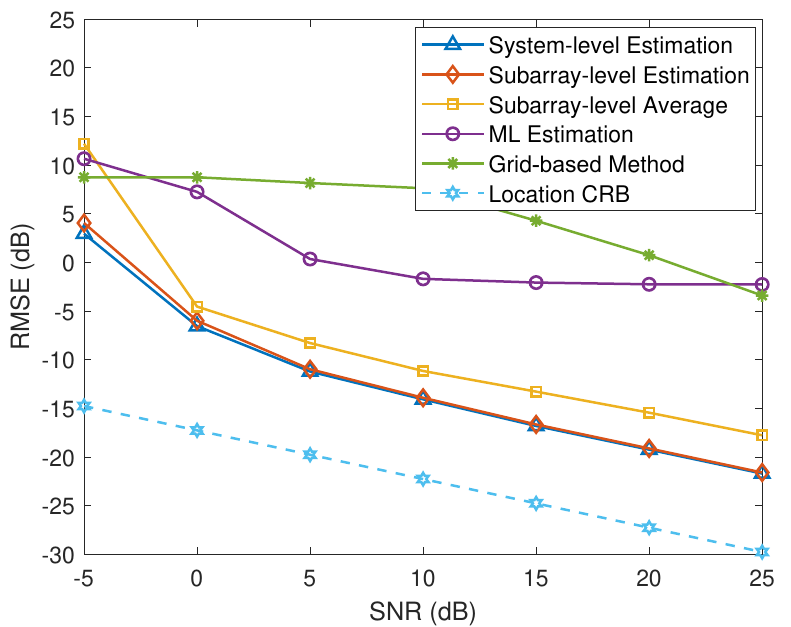}
		\caption{ {\ac{rmse} of location estimation for the proposed  algorithms compared with benchmark schemes versus receive \ac{snr}.}}
		\label{simulation1}
        \vspace{-3mm}
\end{figure}

Fig.~\ref{simulation1} compares the localization accuracy of the proposed system-level and subarray-level algorithms with three benchmark schemes. 
Both proposed algorithms demonstrate superior performance over all baselines across the entire \ac{snr} range, achieving comparable estimation errors that asymptotically approach the \ac{crb} at high \ac{snr}s. 
The performance improvement originates from an adaptive reliability weighting mechanism governed by concentration parameters ${\kappa }_{{{\theta}_{m}}\to {{p}_{{{\theta }_{m}}}}}$ and ${\kappa }_{{{\phi }_{n}}\to {{p}_{{{\phi }_{n}}}}}$ in \eqref{VM_theta_m} and \eqref{VM_phi_n}, which effectively suppresses low-quality subarray measurements.
This contrasts with the subarray-averaging scheme that naively combines estimates with equal weights, thereby amplifying noise impacts particularly at low-to-moderate \ac{snr}s.
The system-level approach marginally outperforms the subarray-level counterpart by exploiting cross-subarray information. 
Compared to conventional approaches, our variational framework fundamentally overcomes following limitations:
1) The ML Estimation suffers from local minima in high-dimensional optimization;
 2) Grid search method suffer from discretization errors and cascaded error propagation, whereas the proposed methods  leverage \ac{vbi} to provide continuous domain modeling, thereby eliminating grid mismatches and enhancing estimation accuracy.
These advantages are particularly pronounced in challenging low-\ac{snr} scenarios where the proposed reliability-aware fusion strategies demonstrate robust noise suppression capabilities.

\begin{figure*} [t]
\setlength{\abovecaptionskip}{-0.1cm}
  \centering
 \begin{minipage}[t]{0.31\textwidth}
       \centering		\includegraphics[width=0.9\textwidth]{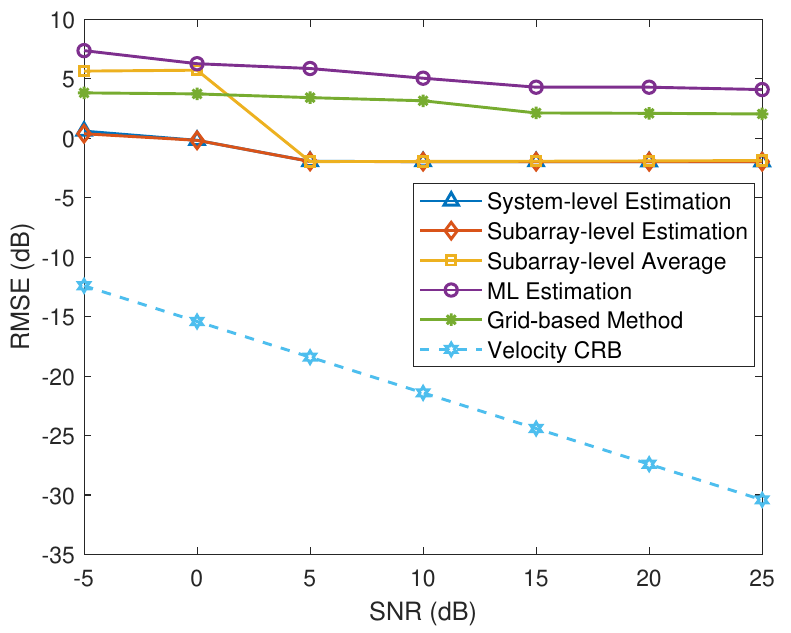}
		\caption{ {\ac{rmse} of velocity estimation for the proposed algorithms compared with benchmark schemes versus receive \ac{snr}.}}
		\label{simulation2}
\end{minipage}
 \begin{minipage}[t]{0.31\textwidth}
      \centering		\includegraphics[width=0.9\textwidth]{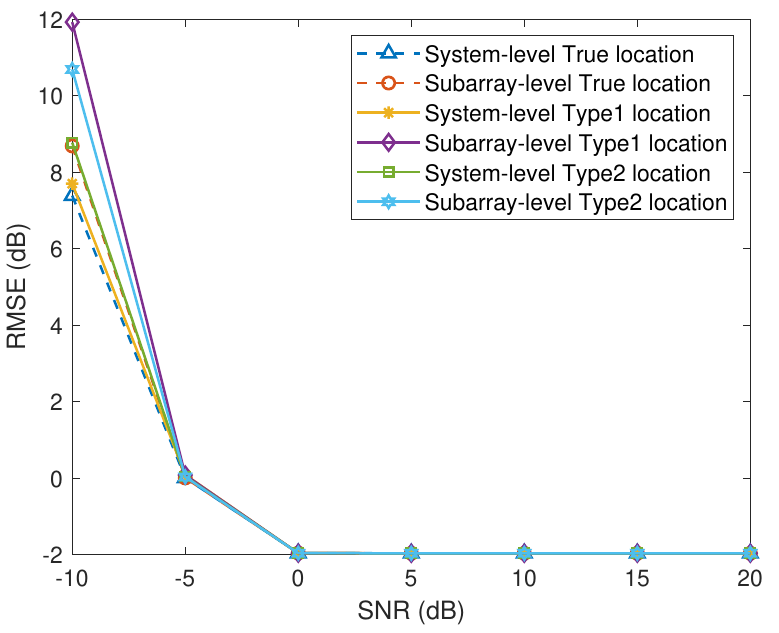}
		\caption{Impact of location estimation accuracy on velocity estimation performance.}
		\label{simulation3}
\end{minipage}
 \begin{minipage}[t]{0.31\textwidth}
\centering		\includegraphics[width=0.9\textwidth]{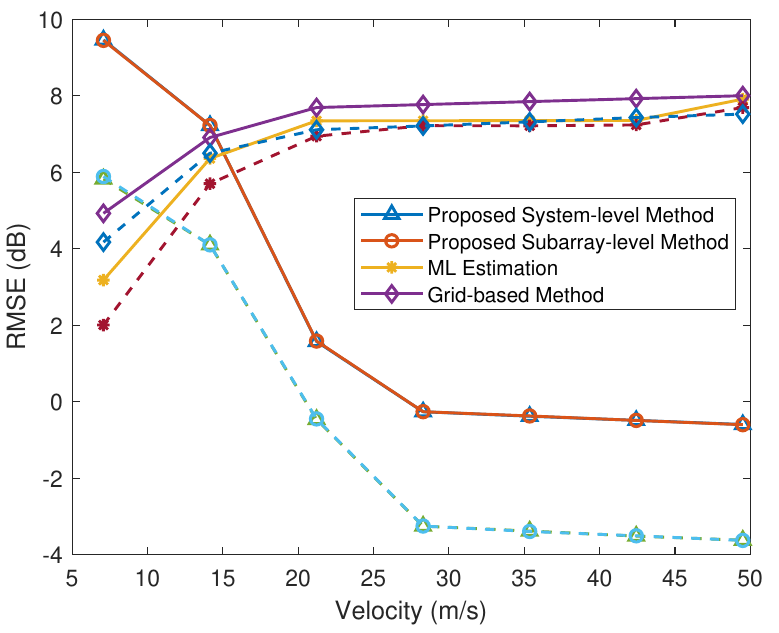}
		\caption{ Impact of target mobility on velocity estimation performance at $ \text{SNR} = 0 \text{dB}$.}
		\label{simulation4}
\end{minipage}
        \vspace{-3mm}
\end{figure*}
Fig.~\ref{simulation2} compares the velocity estimation performance of
the proposed system-level and subarray-level algorithms with three benchmark schemes.
Both system-level and subarray-level variants achieve comparable performance, consistently outperforming all baselines.
The ML Estimator exhibits the worst accuracy due to its requirement for joint optimization over coupled position-velocity parameters, which creates complex local optima in high-dimensional space.
The grid-based method shows discretization-induced error floor, which can be mitigated through our continuous parameter space modeling via \ac{vbi}. 
As the SNR increases, the performance gaps between subarray-averaging and proposed methods narrows. 
This occurs because high-SNR conditions render subarray measurements sufficiently reliable for naive averaging.
However, a discrepancy remains between these converged \ac{rmse} values and the \ac{crb} for velocity estimation. 
This may originate from theoretical constraints of variational approximation, while the mean-field assumption enables tractable inference, it introduces minor information loss in modeling parameter correlations.

Fig.~\ref{simulation3}  demonstrates the impact of location estimation accuracy on velocity estimation, where "True" uses the actual location, while "Type 1" and "Type 2" employ system-level and subarray-level estimated locations, respectively.
At low SNRs, the subarray-level method's velocity estimation suffers considerably, and the choice of location estimation method (True, Type 1, or Type 2) significantly affects the results. 
This degradation stems from error amplification through the bistatic Doppler coupling in \eqref{fmn_geomitric_relation}, causing  biased gradient calculations in velocity estimation.
In contrast, the proposed system-level method can closely match the performance achieved when using the true location, regardless of the location estimation type (Type 1 or Type 2). 
This is due to the system-level approach's ability to coherently fuse information from all subarray pairs, mitigating the impact of individual subarray estimation errors.
As SNR increases, all methods converge  to similarly low \ac{rmse} values, validating the asymptotic optimality of our approach.

 {Fig.~\ref{simulation4} illustrates the impact of target mobility on the velocity estimation performance at an \ac{snr} of $0 $ dB for \ac{cpi} lengths $L=400$ (solid lines) and $L=800$ (dashed lines).}
In these simulations, the target's location and velocity direction remain constant, while the velocity magnitude varies.
It is observed that the extended \ac{cpi} operation enhances the precision of velocity estimation, and the performance gap between the proposed and conventional methods progressively increases with velocity, revealing the inherent advantage of our architecture in high-dynamic scenarios. 
This velocity resilience stems from two aspects: the high-velocity target induce stronger Doppler diversity across spatially distributed subarrays, creating distinctive frequency signatures that our variational framework effectively decouples and exploits, and the hierarchical message passing architecture inherently compensates for spatial-Doppler coupling through adaptive reliability weighting of subarray measurements.

\begin{figure*} [t]
\setlength{\abovecaptionskip}{-0.1cm}
\centering
\begin{minipage}[t]{0.31\textwidth}
\centering
\includegraphics[width=0.9\textwidth]{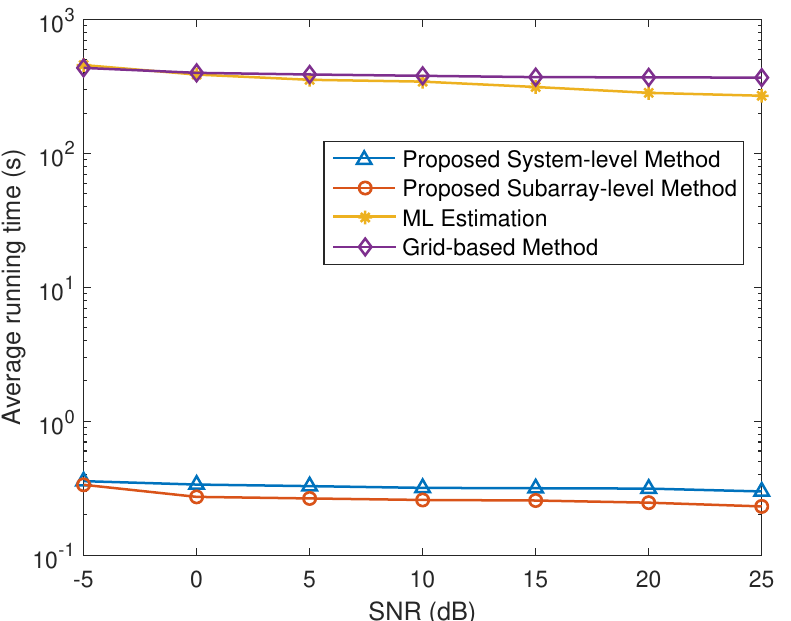}
		\caption{ Average running times for the proposed  algorithms compared with benchmark schemes versus receive \ac{snr}.}
		\label{simulation5}
\end{minipage}	
 \begin{minipage}[t]{0.31\textwidth}
        \centering		\includegraphics[width=0.9\textwidth]{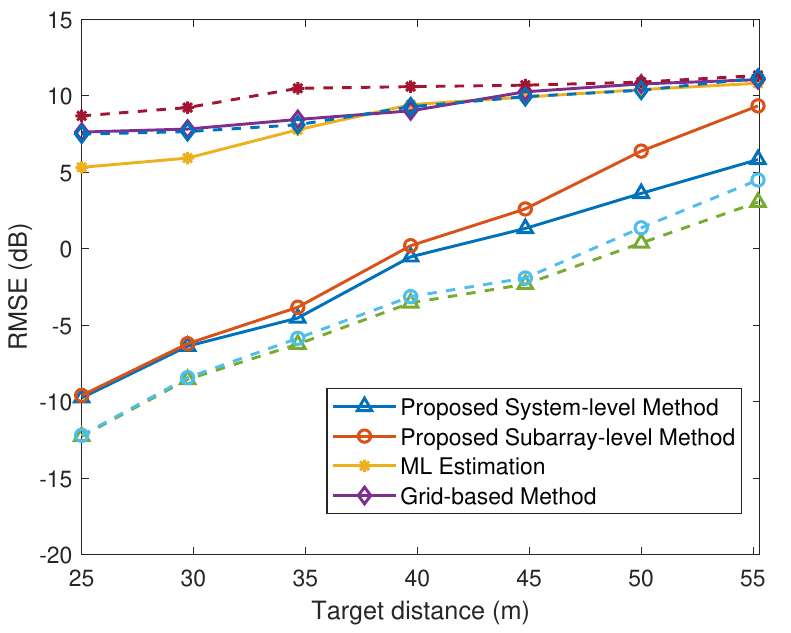}
		\caption{ Impact of target distance on location estimation performance with $M=32$ and $M=64$.}
		\label{simulation6a}
\end{minipage}
 \begin{minipage}[t]{0.31\textwidth}
        \centering		\includegraphics[width=0.9\textwidth]{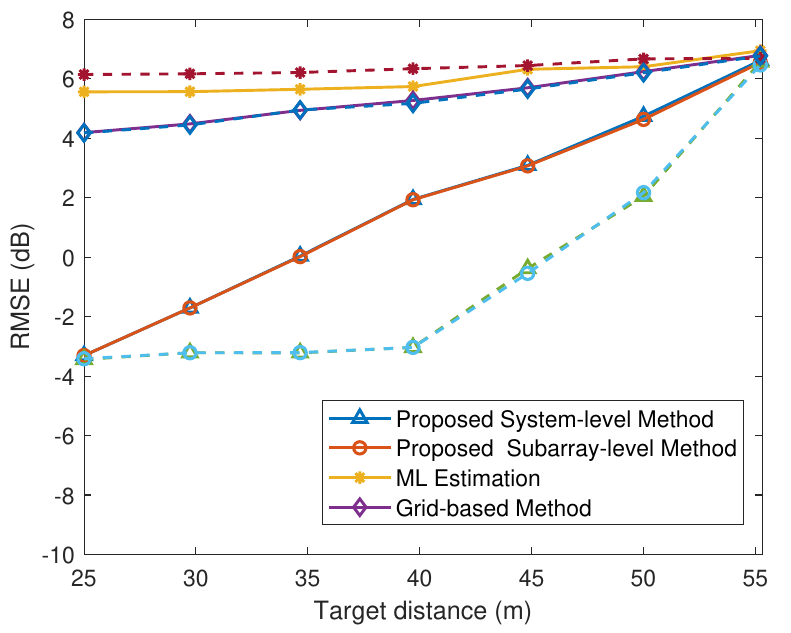}
		\caption{ \ac{rmse} of location estimation for the proposed  algorithms compared with benchmark schemes versus receive \ac{snr}.}
		\label{simulation6b}
\end{minipage}
\vspace{-3mm}
\end{figure*}

Fig.~\ref{simulation5} illustrates the average running time comparison between the proposed algorithms and benchmark schemes across different SNRs.
The computational complexities are $\mathcal{O}(N_{\text{vi}}K_t K_rM(M+L))$ for both system-level and subarray-level proposed algorithms, $\mathcal{O}(N_{\text{iter}}{K_t^2}{K_r^2}M^3L^3)$ for ML estimation, and $\mathcal{O}(K_t K_r(M^6L^3))$ for the grid-based method, where $I$ denotes the number of grid points per parameter dimension. The proposed algorithms achieve $10^{-1}$-$10^0$ seconds average running times, substantially outperforming the ML estimation and the grid-based method ($10^2$-$10^3$ seconds). This superior efficiency can be attributed to two key factors: the parallel processing capability during the variational inference stage, where parameter estimation for each subarray pair can be computed simultaneously, and the well-structured computational framework that avoids the exponential complexity of grid search and the quadratic dependency on array dimensions of ML estimation.


 {Fig.~\ref{simulation6a} illustrates the \ac{rmse} of location estimation  for subarray sizes $M=32$ (solid lines) and $M=64$ (dashed lines), with fixed velocity and total number of transmit and receive antennas.}
It can be observed that the accuracy of location estimation for for all schemes decreases as the distance increases, which is due to the diminishing curvature of the spherical wavefront as the target located farther away from the \ac{bs}. 
The performance gap between the benchmark schemes and the proposed methods diminishes with increasing distance due to the gradual weakening of the near-field effect, as the signal propagation becomes more similar to the far-field case.
For the two proposed approaches, the location estimation performance with $M=64$ is consistently better than that with $M=32$ for each method across all distances.
This improvement can be attributed to the more accurate estimation of the \ac{doa} and \ac{dod} by larger subarrays, which provide a higher angular resolution and are less sensitive to noise, leading to more precise location estimates after fusion.


 {Fig.\ref{simulation6b} illustrates the \ac{rmse} performance of velocity estimation versus target distance for the proposed algorithms and benchmark schemes for subarray sizes $M=32$ (solid lines) and $M=64$ (dashed lines).}
It can be seen that the proposed system-level and subarray-level methods exhibit similar performance across the distance range considered, outperforming the ML estimation and grid-based method. 
As the target distance increases, the performance gap between the proposed algorithms and the benchmarks diminishes. 
Similar to Fig. 7, the velocity estimation accuracy improves when $M=64$ compared to $M=32$, due to the alternating optimization of Doppler, \ac{dod}, and \ac{doa} for each subarray during the variational inference process. 
Larger subarrays enable higher angular estimation precision, which consequently improves the Doppler estimation accuracy. 

\section{Conclusion}\label{section_conclusion}
	    In this paper, a novel subarray-based \ac{vmp} algorithm has been proposed for joint near-field location and velocity estimation.
        Simulation results have demonstrated its superior performance compared to existing methods, achieving centimeter-level location accuracy and sub-m/s velocity accuracy with significantly reduced computational complexity.
        It has been found that increasing the \ac{cpi} length or subarray size generally leads to improved estimation accuracy. 
        Moreover, the proposed method has exhibited robust performance for the target with high mobility.
        Therefore, the low complexity and real-time capability of the proposed algorithm make it highly suitable for implementation in real-world near-field sensing scenarios.

\appendices
\section{Proof of Lemma~\ref{lemma_subarray_vm}}\label{proof_of_lemma_subarray_vm}
By substituting  the joint \ac{pdf} \eqref{joint_pdf_trsubarray}
 and postulated \ac{pdf} \eqref{surrogate_posteri_pdf} into \eqref{trans_elbo_wrt_qu}, we  obtain
\begin{equation} \label{trans_elbo_wrt_qphi_n}
   \setlength{\abovedisplayskip}{2pt}
	\setlength{\belowdisplayskip}{1pt}
\begin{aligned}
        \ln q( {{\phi }_{n}} | {{\mathbf{z}}_{mn}})\!=\!{{E}_{\backslash {{\phi }_{n}}}}\![ \ln p( {{\mathbf{z}}_{mn}}| \boldsymbol{\Theta} ;\boldsymbol{\alpha }) ]\!+\!\ln p( {{\phi }_{n}} )\!+\!\text{const}.
\end{aligned}
\end{equation}
By substituting the likelihood \ac{pdf} \eqref{likelihood_zmn} into \eqref{trans_elbo_wrt_qphi_n} and retaining only the terms dependent on $\phi_n$, we further obtain
\begin{equation}
   \setlength{\abovedisplayskip}{2pt}
	\setlength{\belowdisplayskip}{1pt}
\begin{aligned}
     \ln q ( {{\phi }_{n}} | {{\mathbf{z}}_{mn}} )\!=\!
     -\frac{2}{\sigma }{{E}_{\backslash {{\phi }_{n}}}} [ \Re \{ \mathbf{z}_{mn}^{H} ( \mathbf{a} ( {{\phi }_{n}}  )\otimes {{\mathbf{c}}_{mn} } )  \}  ] \!+\!\ln p ( {{\phi }_{n}}  ),
\end{aligned}
\end{equation}
By carrying out the required expectations, we finally obtain
\begin{equation}\label{posteri_pdf_phin_exp}
   \setlength{\abovedisplayskip}{2pt}
	\setlength{\belowdisplayskip}{2pt}
    q( {{\phi }_{n}}| {{\mathbf{z}}_{mn}})\propto p( {{\phi }_{n}})\exp\{ \Re( (\boldsymbol{\eta }_{mn} )^{H}( \mathbf{a}( {{\phi }_{n}})\otimes {{{\mathbf{\hat{c}}}}_{mn} }) ) \}.
\end{equation}
xpressing $\mathbf{a}( {{\phi }_{n}})$ in terms of the antenna indices, we can rewrite \eqref{posteri_pdf_phin_exp} as
\begin{equation}\label{posteri_pdf_phin_exp_mfold}
   \setlength{\abovedisplayskip}{2pt}
	\setlength{\belowdisplayskip}{2pt}
    q( {{\phi }_{n}} | {{\mathbf{z}}_{mn}})\propto p( {{\phi }_{n}})\prod\nolimits_{k=0}^{M-1}{\exp\{ \Re ( \tilde{\eta }_{\phi_n, k}^{*}{{e}^{jk{{\phi }_{n}}}})\}},
\end{equation}
where $\tilde{\eta }_{\phi_n, k}^{*}$ has the polar form ${\tilde{\eta }_{\phi_n, k}}={{\kappa }_{\phi_n, k}}{{e}^{j{{\mu }_{\phi_n, k}}}}$.
When $k=0$, the corresponding factor in \eqref{posteri_pdf_phin_exp_mfold} is a constant and can be removed. 
The factors with index $k\geq 1$ have the form of a $k$-fold wrapped \ac{vm} distribution $\mathcal{M}(k\phi_n, \tilde{\eta}_{\phi_n, k})$, which can be approximated by a mixture of \ac{vm} distributions\cite{badiu2017variational}.

To fully exploit the multiplicative closure property of \ac{vm} distributions, we also model the prior distribution of $\phi_n$ using a \ac{vm} distribution, i.e., $p(\phi_n) = \mathcal{M}(\phi_n; \bar{\eta}_{\phi_n})$. 
We define the function to represent the exponent of \eqref{posteri_pdf_phin_exp}, which is given by 
$g( {{\phi }_{n}})=\Re ( \bar{\eta }_{\phi_n}^{*}{{e}^{j{{\phi }_{n}}}}+\sum\nolimits_{k=1}^{M-1}{\eta _{\phi_n, k}^{*}{{e}^{jk{{\phi }_{n}}}}} )$.

According to the properties of wrapped normal distributions and their similarity to \ac{vm} distributions\cite{mardia2009directional}, we can obtain the \ac{vm} approximation of $q( {{\phi}_{n}}| {{\mathbf{z}}_{mn}})$ as $\mathcal{M}(\phi_n, \hat{\eta}_{\phi_n})$, where $\hat{\eta}_{\phi_n} = \hat{\kappa}_{\phi_n}e^{j{\hat{\phi}_n}}$.
The mean direction is ${{\hat{\phi }}_{n}}={{\bar{\phi }}_{n}}-\frac{{g}'({{\phi }_{n}})}{{g}''({{\phi}_{n}})}$, where ${{\bar{\phi }}_{n}}$ is obtained at the maximum of $g({{\phi}_{n}})$.
The concentration is given by ${{\hat{\kappa }}_{{{\phi }_{n}}}}={{A}^{-1}}(\exp(\frac{1}{2{g}''({{\phi }_{n}})}) )$.

The proof follows an identical variational structure for parameters $\theta_m$ and $\tilde{f}_{mn}$.

\section{Proof of Lemma~\ref{lemma_beta_gaussian}}\label{proof_of_lemma_beta_mn}
The optimal form of $q( {{\beta}_{mn}}| {{\mathbf{z}}_{mn}} )$ can be obtained as
\begin{equation} \label{trans_elbo_wrt_beta_mn}
   \setlength{\abovedisplayskip}{2pt}
	\setlength{\belowdisplayskip}{2pt}
    \ln q( {{\beta}_{mn}}| {{\mathbf{z}}_{mn}})\!=\!{{E}_{\backslash {{\beta}_{mn}}}}\!\left[ \ln p( {{\mathbf{z}}_{mn}} | \boldsymbol{\Theta};\boldsymbol{\alpha } ) \right]\!+\!\ln p( {\beta}_{mn})\!+\!\text{const}.
\end{equation}
By substituting the likelihood \ac{pdf} \eqref{likelihood_zmn} and prior distribution $ \mathcal{CN}(0, \varsigma_{mn})$ into \eqref{trans_elbo_wrt_beta_mn}, we further obtain
\begin{equation}\label{vi_beta_mn_exp_pdf}
   \setlength{\abovedisplayskip}{2pt}
	\setlength{\belowdisplayskip}{2pt}
\ln q(\beta_{mn}|\mathbf{z}_{mn})\!\!\propto \!\!-\!\frac{2}{\sigma}\!\Re(\beta_{mn}\boldsymbol{\mu}_{mn}^H\mathbf{z}_{mn}) \!\!-\!\! \left(\!\!\frac{1}{\varsigma_{mn}}\!\! + \!\!\frac{2}{\sigma}\|\boldsymbol{\mu}_{mn}\|_2^2\!\right)\!|\beta_{mn}|^2
\end{equation}
Recognizing this as the exponent of a complex normal distribution $\mathcal{CN}(\hat{\beta}_{mn}, \tilde{\varsigma}_{mn})$, we identify
\begin{equation}
   \setlength{\abovedisplayskip}{3pt}
	\setlength{\belowdisplayskip}{3pt}
 {{\hat{\beta }}_{mn}}\!=\!\frac{{{\varsigma }_{mn}}\boldsymbol{\mu }_{mn}^{H}{{\mathbf{z}}_{mn}}}{\sigma \!+\!{{\varsigma }_{mn}}\| {{\boldsymbol{\mu }}_{mn}} \|_{2}^{2}},\  
  {{\tilde{\varsigma}}_{mn}}\!=\!\frac{{{\varsigma }_{mn}}\sigma }{2( \sigma \!+\!{{\varsigma }_{mn}}\| {{\boldsymbol{\mu }}_{mn}} \|_{2}^{2} )},
\end{equation}
which completes the proof.

\section{Proof of Lemma~\ref{lemma_hyper_est}}\label{proof_of_lemma_hyperparameter}
With  $q\left( \boldsymbol{\Theta} \left| {{\mathbf{z}}_{mn}}\right. \right)$ fixed, the \ac{elbo} in \eqref{elbo_origi_express} can be expressed as a function of the hyperparameters $\boldsymbol{\alpha}$, which is given by
\begin{equation}\label{elbo_hyper_expln}
   \setlength{\abovedisplayskip}{2pt}
	\setlength{\belowdisplayskip}{2pt}
    \mathcal{L}( \boldsymbol{\alpha })={{E}_{q( \mathbf{\Theta }| {{\mathbf{z}}_{mn}}  )}}[ \ln p( {{\mathbf{z}}_{mn}}|\Theta ;\boldsymbol{\alpha } )+\ln p( {{\beta }_{mn}};\boldsymbol{\alpha } ) ]+\text{const}.
\end{equation}
Substituting the likelihood \ac{pdf} \eqref{likelihood_zmn} and the prior distribution of $\beta_{mn}$ into \eqref{elbo_hyper_expln}, we obtain
\begin{equation}
   \setlength{\abovedisplayskip}{2pt}
	\setlength{\belowdisplayskip}{2pt}
    \begin{aligned}
       & \mathcal{L}( \boldsymbol{\alpha })=-{{M}^{2}}L\ln \sigma -\ln {{\varsigma }_{mn}}-{2}/{{{\varsigma }_{mn}}}( {{{\tilde{\varsigma}}}_{mn}}+{| {{{\hat{\beta }}}_{mn}}|^{2}} ) -\\
        &{2}/{\sigma }(\| {{\mathbf{z}}_{mn}} \|_{2}^{2}+2\Re( {{{\hat{\beta }}}_{mn}}\mathbf{z}_{mn}^{H}{\hat{\boldsymbol{\mu}}_{mn}} )+\| {\hat{\boldsymbol{\mu }}_{mn}} \|_{2}^{2}( {{{\tilde{\varsigma}}}_{mn}}+{|{{{\hat{\beta }}}_{mn}}|^2}) ).
    \end{aligned}
\end{equation}
Closed-form hyperparameter estimates in \eqref{hyper_est_total} are obtained by setting $\partial\mathcal{L}/\partial\sigma = 0$ and $\partial\mathcal{L}/\partial\varsigma_{mn} = 0$, leveraging the moments $\mathbb{E}[\beta_{mn}] = \hat{\beta}_{mn}$ and $\mathbb{E}[|\beta_{mn}|^2] = \tilde{\varsigma}_{mn} + |\hat{\beta}_{mn}|^2$.

\section{Gaussian Message Approximation in (\ref{gaussian_approx_p0_pdf})} 
\label{calculate_p0_gaussian_pdf}
The expression of $\prod\nolimits_{m=1}^{{{K}_{t}}}{\prod\nolimits_{n=1}^{{{K}_{r}}}{{{\Delta }_{{{p}_{{{\theta }_{m}}}}\to {{\mathbf{p}}_{0}}}}( {{\mathbf{p}}_{0}} ){{\Delta }_{{{p}_{{{\phi }_{n}}}}\to {{\mathbf{p}}_{0}}}}( {{\mathbf{p}}_{0}} )}}$ is given by
\begin{equation}\label{message_input_p0}
   \setlength{\abovedisplayskip}{2pt}
	\setlength{\belowdisplayskip}{2pt}
    \begin{aligned}
  & \prod\nolimits_{m=1}^{{{K}_{t}}}{\prod\nolimits_{n=1}^{{{K}_{r}}}{{{\Delta }_{{{p}_{{{\theta }_{m}}}}\to {{\mathbf{p}}_{0}}}}( {{\mathbf{p}}_{0}}){{\Delta }_{{{p}_{{{\phi }_{n}}}}\to {{\mathbf{p}}_{0}}}}( {{\mathbf{p}}_{0}})}} \\ 
 &\propto \!\exp\!\big\{ \sum\nolimits_{m=1}^{Kt}\!{\sum\nolimits_{n=1}^{{{K}_{r}}}\!{{{\kappa }_{{{\theta }_{m}}\to {{p}_{{{\theta }_{m}}}}}}\!\!\cos (\chi {{( \mathbf{e}_{m}^{t} )^{T}}}\!{{\mathbf{e}}_{x}}\!\!-\!{{\mu }_{{{\theta }_{m}}\to {{p}_{{{\theta }_{m}}}}}} )}} \\ 
 &  +{{\kappa }_{{{\phi }_{n}}\to {{p}_{{{\phi }_{n}}}}}}\cos (  \chi {{( \mathbf{e}_{n}^{r})^{T}}}{{\mathbf{e}}_{x}}-{{\mu }_{{{\phi }_{n}}\to {{p}_{{{\phi }_{n}}}}}} ) \big\},
\end{aligned}
\end{equation}
where $\chi\triangleq\frac{2\pi d}{\lambda }$,  $\mathbf{e}_{m}^{t}\triangleq\frac{{{\mathbf{p}}_{0}}-\mathbf{p}_{m}^{t}}{{{\left\| {{\mathbf{p}}_{0}}-\mathbf{p}_{m}^{t} \right\|}_{2}}}$, and $\mathbf{e}_{n}^{r}\triangleq\frac{{{\mathbf{p}}_{0}}-\mathbf{p}_{n}^{r}}{{{\left\| {{\mathbf{p}}_{0}} -\mathbf{p}_{n}^{r}\right\|}_{2}}}$.

We resort to the gradient descent method to find the local maximum of \eqref{message_input_p0} which is used as the mean vector ${{\mathbf{m}}_{\mathcal{G}}}$ of the approximated Gaussian message in \eqref{gaussian_approx_p0_pdf}. 
The covariance matrix ${{\mathbf{C}}_{\mathcal{G}}}$ of the approximated Gaussian message is given by the Hessian matrix at ${{\mathbf{p}}_{0}}={{\mathbf{m}}_{\mathcal{G}}}$.
We denote the exponential term of \eqref{message_input_p0} as  $f_p({{\mathbf{p}}_{0}})$. 
By defining $\mathbf{w}_{m}^{t}\triangleq\frac{\partial {{( \mathbf{e}_{m}^{t} )^{T}}}{{\mathbf{e}}_{x}}}{\partial {{\mathbf{p}}_{0}}}=\frac{-{{\mathbf{e}}_{x}}+{{( \mathbf{e}_{m}^{t})^{T}}}{{\mathbf{e}}_{x}}{{ \mathbf{e}_{m}^{t} }}}{{{\| \mathbf{p}_{m}^{t}-{{\mathbf{p}}_{0}}\|}_{2}}}$, and $ \mathbf{w}_{n}^{r}\triangleq \frac{\partial {{( \mathbf{e}_{n}^{r})^{T}}}{{\mathbf{e}}_{x}}}{\partial {{\mathbf{p}}_{0}}}=\frac{-{{\mathbf{e}}_{x}}+{{( \mathbf{e}_{n}^{r})^{T}}}{{\mathbf{e}}_{x}}{{\mathbf{e}_{n}^{r}}}}{{{\| \mathbf{p}_{n}^{r}-{{\mathbf{p}}_{0}} \|}_{2}}}$, the gradient of $f_p({{\mathbf{p}}_{0}})$ is derived as
\begin{equation}
   \setlength{\abovedisplayskip}{4pt}
	\setlength{\belowdisplayskip}{4pt}
    \frac{\partial {{f}_{p}}( {{\mathbf{p}}_{0}})}{\partial {{\mathbf{p}}_{0}}}=\sum\nolimits_{m=1}^{Kt}{\sum\nolimits_{n=1}^{{{K}_{r}}}{{\boldsymbol{\vartheta}}_{m}^{t}+{\boldsymbol{\vartheta}}_{n}^{r}}},
\end{equation}
with  ${\boldsymbol{\vartheta}}_{m}^{t}\triangleq{- \chi {{\kappa }_{{{\theta }_{m}}\to {{p}_{{{\theta }_{m}}}}}}\sin ( \chi {{( \mathbf{e}_{m}^{t})^{T}}}{{\mathbf{e}}_{x}}-{{\mu }_{{{\theta }_{m}}\to {{p}_{{{\theta }_{m}}}}}} )}\mathbf{w}_{m}^{t}$, and  $ {\boldsymbol{\vartheta}}_{n}^{r}\triangleq- \chi {{\kappa }_{{{\phi }_{n}}\to {{p}_{{{\phi }_{n}}}}}}\sin (  \chi {{( \mathbf{e}_{n}^{r})^{T}}}{{\mathbf{e}}_{x}}-{{\mu }_{{{\phi }_{n}}\to {{p}_{{{\phi }_{n}}}}}})\mathbf{w}_{n}^{r}$.
Therefore, we iteratively update $\mathbf{p}_0$ according to the following rule:
\begin{equation} \label{p0_sys_level_update}
   \setlength{\abovedisplayskip}{4pt}
	\setlength{\belowdisplayskip}{3pt}
\mathbf{p}_{0}^{(i+1)}=\mathbf{p}_{0}^{(i)}+{{\delta }_{p}}{{\left. \frac{\partial {{f}_{p}}( {{\mathbf{p}}_{0}} )}{\partial {{\mathbf{p}}_{0}}} \right|}_{{{\mathbf{p}}_{0}}=\mathbf{p}_{0}^{(i)}}},
\end{equation}
where ${{\delta }_{p}}$ is the step size and $i$ is the iteration index. The iteration stops when $\mathbf{p}_0$ reaches a local optimum, which is obtained as ${{\mathbf{m}}_{\mathcal{G}}}$.
The covariance matrix ${{\mathbf{C}}_{\mathcal{G}}}$ is obtained by calculating the Hessian matrix of ${{f}_{p}}( {{\mathbf{p}}_{0}} )$ at ${{\mathbf{m}}_{\mathcal{G}}}$.
The  Hessian matrix of ${{f}_{p}}( {{\mathbf{p}}_{0}} )$ is obtained as
\begin{equation} \label{hessian_p0_central}
   \setlength{\abovedisplayskip}{5pt}
	\setlength{\belowdisplayskip}{1pt}
    \frac{{{\partial }^{2}}{{f}_{p}}( {{\mathbf{p}}_{0}} )}{\partial {{\mathbf{p}}_{0}}\partial \mathbf{p}_{0}^{T}}=\sum\nolimits_{m=1}^{Kt}{\sum\nolimits_{n=1}^{{{K}_{r}}}{\mathbf{J}_{m}^{t}+\mathbf{J}_{n}^{r}}},
\end{equation}
with
\begin{equation}
   \setlength{\abovedisplayskip}{1pt}
	\setlength{\belowdisplayskip}{3pt}
    \begin{aligned}
  & \mathbf{J}_{m}^{t}\!=\!-{{ \chi }^{2}}{{\kappa }_{{{\theta }_{m}}\to {{p}_{{{\theta }_{m}}}}}}\!\!\cos(\chi {\gamma_{m,x}}\!\!-\!{{\mu }_{{{\theta }_{m}}\to {{p}_{{{\theta }_{m}}}}}}\!)\mathbf{w}_{m}^{t}{{( \mathbf{w}_{m}^{t} )^{T}}} \\ 
 & - \chi {{\kappa }_{{{\theta }_{m}}\to {{p}_{{{\theta }_{m}}}}}}\sin (  \chi \gamma_{m,x}-{{\mu }_{{{\theta }_{m}}\to {{p}_{{{\theta }_{m}}}}}})\frac{\partial \mathbf{w}_{m}^{t}}{\partial \mathbf{p}_{0}^{T}},
\end{aligned}
\end{equation}
where  
$\frac{\partial \mathbf{w}_{m}^{t}}{\partial \mathbf{p}_{0}^{T}}=\frac{-2{{\mathbf{e}}_{x}}{{( \mathbf{e}_{m}^{t})^{T}}}+3\gamma_{m,x}\mathbf{e}_{m}^{t}{{( \mathbf{e}_{m}^{t} )^{T}}}-\gamma_{m,x}\mathbf{I}}{{{\| \mathbf{p}_{m}^{t}-{{\mathbf{p}}_{0}}\|}_{2}}}$, ${\gamma_{m,x}^t\triangleq( \mathbf{e}_{m}^{t} )^{T}}{{\mathbf{e}}_{x}}$, and ${\gamma_{n,x}^r\triangleq( \mathbf{e}_{n}^{r} )^{T}}{{\mathbf{e}}_{x}}$.
The matrix $\mathbf{J}_{n}^{r}$ can be calculated similarly.

Therefore, the covariance matrix ${{\mathbf{C}}_{\mathcal{G}}}$ is obtained as the negative inverse of the Hessian matrix of ${{f}_{p}}\left( {{\mathbf{p}}_{0}} \right)$ at ${{\mathbf{m}}_{\mathcal{G}}}$, which is given by
\begin{equation} \label{covar_p0_sys_level}
   \setlength{\abovedisplayskip}{2pt}
	\setlength{\belowdisplayskip}{2pt}
{{\mathbf{C}}_{\mathcal{G}}^{-1}}={{ -{{\left. \frac{\partial {{f}_{p}}( {{\mathbf{p}}_{0}} )}{\partial {{\mathbf{p}}_{0}}\partial \mathbf{p}_{0}^{T}} \right|}_{{{\mathbf{p}}_{0}}={{\mathbf{m}}_{\mathcal{G}}}}} }}.
\end{equation}

\section{Gaussian Message Approximation in (\ref{gaussian_approx_dis_p0_pdf})} 
\label{calculate_dis_p0_gaussian_pdf}
Based on \eqref{message_input_p0}, the condition for the function $\ln({{\Delta }_{{{p}_{{{\theta }_{m}}}}\to {{\mathbf{p}}_{0}}}}( {{\mathbf{p}}_{0}}){{\Delta }_{{{p}_{{{\phi }_{n}}}}\to {{\mathbf{p}}_{0}}}}( {{\mathbf{p}}_{0}}))$ to reach its maximum value is $ \sin {{\tilde{\theta }}_{m}}=\frac{{{\mu }_{{{\theta }_{m}}\to {{p}_{{{\theta }_{m}}}}}}}{\chi }$, $\sin {{\tilde{\phi }}_{n}}=\frac{{{\mu }_{{{\phi }_{n}}\to {{p}_{{{\phi }_{n}}}}}}}{\chi }$.
Therefore, we have $\mathbf{e}_{m}^{t}=\left[ \frac{1}{\chi }{{\mu }_{{{\theta }_{m}}\to {{p}_{{{\theta }_{m}}}}}},\sqrt{1-\frac{1}{{{\chi }^{2}}}\mu _{{{\theta }_{m}}\to {{p}_{{{\theta }_{m}}}}}^{2}} \right]$,  and $\mathbf{e}_{n}^{r}=\left[ \frac{1}{\chi }{{\mu }_{{{\phi }_{n}}\to {{p}_{{{\phi }_{n}}}}}},\sqrt{1-\frac{1}{{{\chi }^{2}}}\mu _{{{\phi }_{n}}\to {{p}_{{{\phi }_{n}}}}}^{2}} \right]$.

By denoting the  distances $\left\| \mathbf{p}_{n}^{r}-{{\mathbf{p}}_{0}} \right\|_2$ and $\left\| \mathbf{p}_{m}^{t}-{{\mathbf{p}}_{0}} \right\|_2$ as $d_n^r$ and $d_m^t$, respectively, we can obtain the bistatic range for $(m,n)$-th \ac{t-r} subarray pair as $d_n^r + d_m^t = c \hat{\tau}_{mn}$, where $\hat{\tau}_{mn}$ is the estimated bistatic time delay derived from the previous matched filtering step.
Therefore,  the maximum can be obtained as ${{\mathbf{m}}_{p}^{(mn)}}=\mathbf{p}_{n}^{r}-d_{n}^{r}\mathbf{e}_{n}^{r}=\mathbf{p}_{m}^{t}-d_{m}^{t}\mathbf{e}_{m}^{t}$, where the bistatic ranges are given by
$d_{n}^{r}=\frac{c{{{\hat{\tau }}}_{mn}}\sqrt{1-\mu _{{{\theta }_{m}}\to {{p}_{{{\theta }_{m}}}}}^{2}/{{\chi }^{2}}}}{\sqrt{1-\mu _{{{\theta }_{m}}\to {{p}_{{{\theta }_{m}}}}}^{2}/{{\chi }^{2}}}+\sqrt{1-\mu _{{{\phi }_{n}}\to {{p}_{{{\phi }_{n}}}}}^{2}/{{\chi }^{2}}}}$ and $d_{m}^{t}=\frac{c{{{\hat{\tau }}}_{mn}}\sqrt{1-\mu _{{{\phi }_{n}}\to {{p}_{{{\phi }_{n}}}}}^{2}/{{\chi }^{2}}}}{\sqrt{1-\mu _{{{\theta }_{m}}\to {{p}_{{{\theta }_{m}}}}}^{2}/{{\chi }^{2}}}+\sqrt{1-\mu _{{{\phi }_{n}}\to {{p}_{{{\phi }_{n}}}}}^{2}/{{\chi }^{2}}}}$, respectively.

Based on \eqref{hessian_p0_central}, the Hessian matrix of the function $\ln({{\Delta }_{{{p}_{{{\theta }_{m}}}}\to {{\mathbf{p}}_{0}}}}( {{\mathbf{p}}_{0}}){{\Delta }_{{{p}_{{{\phi }_{n}}}}\to {{\mathbf{p}}_{0}}}}({{\mathbf{p}}_{0}}))$ is given by ${\mathbf{J}_{p}^{(mn)}={\mathbf{J}_{m}^{t}+\mathbf{J}_{n}^{r}}}$. Therefore, the covariance matrix of the Gaussian distribution can be approximated by the inverse of the Hessian matrix evaluated at the maximum point $\mathbf{m}_p^{(mn)}$:
\begin{equation}
   \setlength{\abovedisplayskip}{4pt}
	\setlength{\belowdisplayskip}{2pt}
  \mathbf{C}_{p}^{(m,n)} \approx -({\mathbf{J}_{p}^{(mn)}})^{-1}\big|_{{{\mathbf{p}}_{0}}=\mathbf{m}_p^{(mn)}},
\end{equation}
with
\begin{equation}
    \begin{aligned}
       \setlength{\abovedisplayskip}{1pt}
	\setlength{\belowdisplayskip}{1pt}
      & \mathbf{J}_{p}^{(mn)}=\frac{{{\chi }^{2}}{{\kappa }_{{{\theta }_{m}}\to {{p}_{{{\theta }_{m}}}}}}}{d_{m}^{t}}  ( 2{{\mathbf{e}}_{x}}{{ ( \mathbf{e}_{m}^{t}  )^{T}}}-3 \gamma_{m,x}^t\mathbf{e}_{m}^{t}{{ ( \mathbf{e}_{m}^{t}  )^{T}}}+\gamma_{m,x}^t\mathbf{I} ) \\ 
 & +{{{\chi }^{2}}{{\kappa }_{{{\phi }_{n}}\to {{p}_{{{\phi }_{n}}}}}}}/{d_{n}^{r}}  ( 2{{\mathbf{e}}_{x}}{{ ( \mathbf{e}_{n}^{r}  )^{T}}}-3\gamma_{n,x}^r\mathbf{e}_{n}^{r}{{ ( \mathbf{e}_{n}^{r}  )^{T}}}+\gamma_{n,x}^r\mathbf{I} ).
    \end{aligned}
\end{equation}

\section{Gaussian Message Approximation in (\ref{gaussian_approx_v0_pdf})} 
\label{calculate_v0_gaussian_pdf}
Based on Definition \ref{definition_vmd}, the expression in \eqref{v0_marg_prob} can be rewritten as 
\begin{equation}\label{vmd_v0_marg_pdf}
       \setlength{\abovedisplayskip}{2pt}
	\setlength{\belowdisplayskip}{1pt}
    \begin{aligned}
        &\prod\nolimits_{m=1}^{{{K}_{t}}}{\prod\nolimits_{n=1}^{{{K}_{r}}}{{{\Delta }_{{{p}_{{{f}_{mn}}}}\to {{\mathbf{v}}_{0}}}}\left( {{\mathbf{v}}_{0}} \right)}}\propto\\
        & \exp \{\sum\nolimits_{m,n}\!{{{\kappa }_{{{f}_{mn}}\to {{p}_{{{f}_{mn}}}}}}\!\cos ( \zeta \mathbf{v}_{0}^{T}{{\mathbf{u}}_{mn}}\!\!-\!\!{{\mu }_{{{f}_{mn}}\to {{p}_{{{f}_{mn}}}}}})} \}.
    \end{aligned}
\end{equation}
Denoting the exponential term in \eqref{vmd_v0_marg_pdf} as a function, we have
\begin{equation}
\begin{aligned}
       \setlength{\abovedisplayskip}{1pt}
	\setlength{\belowdisplayskip}{0.5pt}
    {{f}_{v}}( {{\mathbf{v}}_{0}})\!\triangleq\!\! \sum_{m,n}{{{\kappa }_{{{f}_{mn}}\to {{p}_{{{f}_{mn}}}}}}\!\!\cos( \zeta \mathbf{v}_{0}^{T}{{\mathbf{u}}_{mn}}\!\!-\!\!{{\mu }_{{{f}_{mn}}\to {{p}_{{{f}_{mn}}}}}})}.
    \end{aligned}
\end{equation}
The mean vector ${{\mathbf{m}}_{\mathcal{H}}}$ is obtained at the local maximum of ${{f}_{v}}\left( {{\mathbf{v}}_{0}} \right)$ using the gradient ascent method:
\begin{equation} \label{v0_sys_update}
       \setlength{\abovedisplayskip}{4pt}
	\setlength{\belowdisplayskip}{4pt}
\mathbf{v}_{0}^{(i+1)}=\mathbf{v}_{0}^{(i)}+{{\delta }_{v}}{{\left. \frac{\partial {{f}_{v}}\left( {{\mathbf{v}}_{0}} \right)}{\partial {{\mathbf{v}}_{0}}} \right|}_{{{\mathbf{v}}_{0}}=\mathbf{v}_{0}^{(i)}}},
\end{equation}
where ${{\delta }_{v}}$ is the step size, $i$ is the iteration index, and $\frac{\partial {{f}_{v}}( {{\mathbf{v}}_{0}})}{\partial {{\mathbf{v}}_{0}}}\!\!=\!\!-\!\!\!
        \sum_{m,n}{\zeta {{\kappa }_{{{f}_{mn}}\to {{p}_{{{f}_{mn}}}}}}\!\sin( \zeta \mathbf{v}_{0}^{T}{{\mathbf{u}}_{mn}}\!\!-\!\!{{\mu }_{{{f}_{mn}}\to {{p}_{{{f}_{mn}}}}}}\!){{\mathbf{u}}_{mn}}}$ denotes the gradient of ${{f}_{v}}\left( {{\mathbf{v}}_{0}} \right)$ with respect to ${{\mathbf{v}}_{0}}$.

The covariance matrix ${{\mathbf{C}}_{\mathcal{H}}}$ is obtained as the negative inverse of the Hessian matrix of ${{f}_{v}}\left( {{\mathbf{v}}_{0}} \right)$ evaluated at ${{\mathbf{m}}_{\mathcal{H}}}$:
\begin{equation} \label{v0_sys_level_covar}
       \setlength{\abovedisplayskip}{4pt}
	\setlength{\belowdisplayskip}{2pt}
{{\mathbf{C}}_{\mathcal{H}}^{-1}}= -{{\left. \frac{\partial {{f}_{v}}\left( {{\mathbf{v}}_{0}} \right)}{\partial {{\mathbf{v}}_{0}}\partial \mathbf{v}_{0}^{T}} \right|}_{{{\mathbf{v}}_{0}}={{\mathbf{m}}_{\mathcal{H}}}}} ,
\end{equation}
where 
\begin{equation}\label{hessian_v0_central}
 \setlength{\abovedisplayskip}{4pt}
	\setlength{\belowdisplayskip}{1pt}
    \frac{\partial {{f}_{v}}( {{\mathbf{v}}_{0}})}{\partial {{\mathbf{v}}_{0}}\partial \mathbf{v}_{0}^{T}}\!\!=\!\!-\!\!\!\sum_{m,n}\!{{\zeta}^{2}}\!{{\kappa }_{{{f}_{mn}}\!\to {{p}_{{{f}_{mn}}}}}}\!\!
       \cos(\! \zeta \mathbf{v}_{0}^{T}{{\mathbf{u}}_{mn}}\!\!-\!\!{{\mu }_{{{f}_{mn}}\!\to {{p}_{{{f}_{mn}}}}}}\!){{\mathbf{u}}_{mn}}\mathbf{u}_{mn}^{T}.
\end{equation}

\section{Gaussian Message Approximation in (\ref{gaussian_approx_v0_dis_pdf})} 
\label{calculate_dis_v0_gaussian_pdf}
The  logarithm of the marginal \ac{pdf} in \eqref{vo_marg_pdf_dis_pair} can be expressed as 
\begin{equation}\label{delta_vo_ln}
       \setlength{\abovedisplayskip}{2pt}
	\setlength{\belowdisplayskip}{2pt}
    \begin{aligned}
        \Delta _{{{\mathbf{v}}_{0}}}^{({{\omega }_{i}})}( {{\mathbf{v}}_{0}}) &\propto\!\exp\!\big\{ {{\kappa }_{{{f}_{mn}}\!\to {{p}_{{{f}_{mn}}}}}}\!\!\cos ( \zeta \mathbf{v}_{0}^{T}{{\mathbf{u}}_{mn}}\!\!-\!\!{{\mu }_{{{f}_{mn}}\!\to {{p}_{{{f}_{mn}}}}}}\!)  \\ 
 & +{{\kappa }_{{{f}_{pq}}\to {{p}_{{{f}_{pq}}}}}}\cos ( \zeta \mathbf{v}_{0}^{T}{{\mathbf{u}}_{pq}}-{{\mu }_{{{f}_{pq}}\to {{p}_{{{f}_{pq}}}}}}) \big\}. 
    \end{aligned}
\end{equation}
We define ${{\mathbf{E}}_{{{\omega }_{i}}}}\triangleq [ {{\mathbf{u}}_{mn}},{{\mathbf{u}}_{pq}}] $,  and ${{\boldsymbol{\mu }}_{{{\omega }_{i}}}}\triangleq {{[ {{\mu }_{{{f}_{mn}}\to {{p}_{{{f}_{mn}}}}}},{{\mu }_{{{f}_{pq}}\to {{p}_{{{f}_{pq}}}}}}]}^{T}}$.
Therefore, the mean of the Gaussian distribution is obtained at the maximum of \eqref{delta_vo_ln}, which is achieved when the  equality $\mathbf{v}_{0}^{T}{{\mathbf{E}}_{{{\omega }_{i}}}}={{\boldsymbol{\mu }}_{{{\omega }_{i}}}}/{\zeta }$ holds.
The mean vector is obtained by $\mathbf{m}_{v}^{({{\omega }_{i}})}={{( \mathbf{E}_{{{\omega }_{i}}}^{T} )^{-1}}}\boldsymbol{\mu }_{{{\omega }_{i}}}^{T}/{\zeta }$.
Based on \eqref{hessian_v0_central}, the covariance matrix is given by
\begin{equation}
       \setlength{\abovedisplayskip}{2pt}
	\setlength{\belowdisplayskip}{2pt}
    \begin{aligned}
        \mathbf{C}_{v}^{({{\omega }_{i}})}\!\!=\!-{{( {{\kappa }_{{{f}_{mn}}\!\to {{p}_{{{f}_{mn}}}}}}\!{{\mathbf{u}}_{mn}}\mathbf{u}_{mn}^{T}\!\!+\!{{\kappa }_{{{f}_{pq}}\!\to {{p}_{{{f}_{pq}}}}}}\!{{\mathbf{u}}_{pq}}\mathbf{u}_{pq}^{T} )^{-1}}}\!/{{{\zeta }^{2}}}.
    \end{aligned}
\end{equation}

\bibliographystyle{IEEEtran}
\bibliography{mybib}
\end{document}